\documentclass[aps,prl,twocolumn,superscriptaddress,noeprint,longbibliography, nofootinbib]{revtex4-2}
\usepackage[utf8]{inputenc}
\usepackage{amsmath}
\usepackage{amsthm}
\usepackage{enumitem}
\usepackage{amssymb}
\usepackage{wasysym}
\usepackage{oplotsymbl}
\usepackage{bm} %
\usepackage{dsfont} %
\usepackage{mathrsfs} %
\usepackage{physics}
\usepackage{mathtools}

\usepackage{xcolor}
\usepackage{framed}
\definecolor{darkblue}{rgb}{0,0,.65}
\definecolor{darkgreen}{rgb}{0.28,0.41,0.19}
\definecolor{nicegreen}{rgb}{0.28,0.85,0.19}

\usepackage[colorlinks=true,linkcolor=blue,allcolors=blue]{hyperref}
\usepackage[capitalize,nameinlink]{cleveref}

\def\equationautorefname~#1\null{Eq. (#1)\null}

\newcommand{\appref}[1]{\hyperref[#1]{App.~\ref*{#1}}}

\newcommand{\hilbert}{\ensuremath{\mathscr{H}}}

\newtheorem{prototheorem}{Theorem}[section]

\newtheorem{protolemma}[prototheorem]{Lemma}

\newtheorem{corollary}{Corollary}

\newenvironment{lemma}{\colorlet{shadecolor}{gray!15}\begin{shaded}\begin{protolemma}}
{\end{protolemma}\end{shaded}}

\newtheorem{thm}{Theorem}%

\newtheorem{Def}{Definition}

\begin{document}

\title{Bottlenecks in quantum channels and finite temperature phases of matter }

\author{Tibor Rakovszky}
\affiliation{Department of Physics, Stanford University, Stanford, California 94305, USA}
\affiliation{Department of Theoretical Physics, Institute of Physics,
Budapest University of Technology and Economics, M\H{u}egyetem rkp. 3., H-1111 Budapest, Hungary}
\affiliation{HUN-REN-BME Quantum Error Correcting Codes and Non-equilibrium Phases Research Group,
Budapest University of Technology and Economics,
M\H{u}egyetem rkp. 3., H-1111 Budapest, Hungary}
\author{Benedikt Placke}%
\affiliation{Rudolf Peierls Centre for Theoretical Physics, University of Oxford, Oxford OX1 3PU, United Kingdom}
\author{Nikolas P.\ Breuckmann}
\affiliation{School of Mathematics, University of Bristol, Bristol BS8 1UG, United Kingdom}
\author{Vedika Khemani}
\affiliation{Department of Physics, Stanford University, Stanford, California 94305, USA}

\begin{abstract}
    We prove an analogue of the ``bottleneck theorem'', well-known for classical Markov chains, for Markovian quantum channels. In particular, we show that if two regions (subspaces) of Hilbert space are separated by a region that has very low weight in the channel's steady state, then states initialized on one side of this barrier will take a long time to relax, putting a lower bound on the mixing time in terms of an appropriately defined ``quantum bottleneck ratio''. Importantly, this bottleneck ratio involves not only the probabilities of the relevant subspaces, but also the size of off-diagonal matrix elements between them. For low temperature quantum many-body systems, we use the bottleneck theorem to bound the performance of any quasi-local Gibbs sampler. This leads to a new perspective on thermally stable quantum phases in terms of a decomposition of the Gibbs state into multiple components separated by bottlenecks. As a concrete application, we show rigorously that weakly perturbed commuting projector models with extensive energy barriers (which includes certain classical and quantum expander codes) have exponentially large mixing times. 

\end{abstract}

\maketitle

{\it Introduction.}
Open quantum systems encompass many scenarios of interest, from quantum optics and condensed matter physics to quantum computing and simulation~\cite{breuer2002theory,nielsen2010quantum}. Such systems often evolve towards some unique steady state, and a fundamental problem is to characterize the time scale on which this steady state is reached, known as the \emph{mixing time}. When dynamics is engineered to prepare some particular quantum state of interest, such as ground states or Gibbs states of local Hamiltonians~\cite{verstraete2009quantum,kraus2008preparation,lin2025dissipative} (which can encode solutions to various optimization problems), a short mixing time ensures the efficiency of the protocol. On the other hand, a long mixing time for preparing Gibbs states via local dynamics is characteristic of many ordered phases~\cite{martinelli2004lectures,levin2017markov}. By definition, 
a long mixing time also implies that the system preserves some memory of its initial state;  understanding when this happens is therefore relevant in designing protocols for error correction. For all these reasons, establishing general results on mixing times for wide families of physical dynamics is a task of central importance. 

In the context of classical Markov processes, one standard approach that lower bounds the mixing time is the \emph{bottleneck method}~\cite{levin2017markov}: it involves finding a decomposition of the configuration space into two halves such that both have significant probabilities in the steady state, but the flow induced by the Markov chain between them is small; one can then show that a state initialized in one half takes a long time to relax to the true steady state. A classic application~\cite{levin2017markov} of this result is to Glauber dynamics in the low-temperature ferromagnetic regime of the Ising model, which has a long mixing time due to the bottleneck separating configurations with positive and negative net magnetizations. Generalizing this approach to low-temperature states of generic \emph{quantum} Hamiltonians is one of the main motivations of the present work.

In this work, we generalize the bottleneck approach from classical Markov chains to quantum channels. In quantum systems, the standard proof of the bottleneck theorem~\cite{levin2017markov} is invalidated by the possibility of interference effects that make probabilities non-additive and indeed, one might expect that quantum tunnelling can lead to qualitatively different behavior\footnote{One context where analogs of the bottleneck theorem have been proven is for unitary quantum walks~\cite{aharonov2001quantum,szegedy2004quantum}.}. Here, we show that an appropriate generalization of the bottleneck theorem nevertheless \emph{does hold}: if one can find a subspace such that in order to leave it, any quantum trajectory must pass through a low-weight bottleneck region, then the subspace hosts a long-lived approximate steady state, lower bounding the channel's mixing time in terms of a ``quantum bottleneck ratio'' (\autoref{fig:fig1}). The important difference is that this lower bound involves not only probabilities but also the \emph{coherences} (off-diagonal matrix elements) present in the steady state density matrix.

A natural application of our theorem is to the dissipative preparation of ground states and finite-temperature Gibbs states of many-body systems, which are fundamental for both physical and computational applications~\cite{verstraete2009quantum,poulin2009sampling,temme2011quantum,rall2023thermal}. There has been much interest in this problem, with a number of  recent constructions that achieve this goal using quasi-local operations~\cite{lin2025dissipative,chen2023quantum,chen2311efficient,gilyen2024quantum,ding2024efficient,zhan2025rapid}. We show that our bottleneck theorem can be used to lower bound the mixing time of any local or quasi-local channel that prepares a given target steady state and we argue that this leads to superpolynomially long mixing times in many cases of interest, such as the low temperature regime of certain Ising and toric code models. These results also provide a way of rigorously establishing the existence of passive memories in quantum systems, an idea that we further illustrate by proving exponential mixing times in a family of Hamiltonians related to good classical and quantum error correcting codes. 

Beyond bounding mixing times, our theorem provides a dynamical perspective on finite-temperature phases of matter. The bottleneck ratio formalizes the idea of a \emph{free energy barrier} separating the Gibbs state into multiple distinct components. Our theorem shows that when such barriers are present (as are expected in many non-trivial phases of matter), \emph{any} (quasi-)local dynamics will feature multiple near-degenerate steady states. This provides an analogy with non-trivial zero temperature phases, which are often characterized in terms of ground state degeneracies, and places thermal phases within the broader framework of phases in open quantum systems~\cite{rakovszky2024defining}.

{\it Quantum channels and mixing times.} A quantum channel~\cite{nielsen2010quantum} $\mathcal{M}$ is a super-operator that maps any density matrix $\rho$ over Hilbert space $\hilbert$ to another density matrix $\mathcal{M}[\rho]$. In the following we will make use of the Kraus representation of channels, $\mathcal{M}[\rho] = \sum_i K_i \rho K_i^\dagger$ where $\{K_i\}$ are a set of operators with the property $\sum_i K_i^\dagger K_i = 1\!\!1$. We will be interested in Markovian (memoryless) evolutions~\cite{wolf2008dividing,rivas2010entanglement,rivas2014quantum,breuer2016colloquium,li2018concepts}, where the state at time $t$ depends only on the state at the previous step, as $\rho(t) = \mathcal{M}_t[\rho(t-1)]$, where we allow the channel $\mathcal{M}_t$ itself to also depend on $t$, which could correspond, e.g., to discretizing the evolution with time-dependent Lindbladian. A useful representation of such channels is in terms of \emph{quantum trajectories}~\cite{carmichael2009open,molmer1993monte}, where the Kraus operators correspond to random ``quantum jumps'' in an ensemble of time-evolving \emph{pure} quantum states such that on average they reproduce the evolution of $\rho(t)$. 

Every channel has at least one steady state ${\mathcal{M}[\rho_{\rm ss}] = \rho_{\rm ss}}$~\cite{watrous2018theory}, which is unique if $\mathcal{M}$ has no non-trivial invariant subspaces~\cite{burgarth2013ergodic}. The \emph{mixing time} of $\mathcal{M}$ if defined as the minimal time needed to get close to $\rho_{\rm ss}$:
\begin{equation}
    t_{\text{mix},\varepsilon} = \min\left\{t\geq 0: ||\rho(t) - \rho_{\rm ss}||_1 \leq 2\varepsilon, \, \forall \rho(0)\right\},
\end{equation}
where $||A||_1$ is the trace norm and the cutoff $\varepsilon < 1/2$ measures the trace-distance from the steady state. Note that by $t_{\text{mix}}$ \emph{any} initial state gets close to $\rho_{\rm ss}$, corresponding to a worst case scenario. We will show that $t_{\rm mix}$ can be lower bounded in terms of properties of $\rho_{\rm ss}$.

\begin{figure}
    \centering
    \includegraphics[width=\columnwidth]{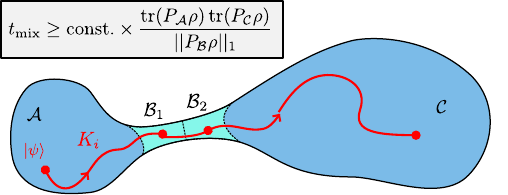}
    \caption{{\bf Quantum bottleneck theorem.} The Hilbert space is split into four parts in such a way that Kraus operators only have matrix elements between neighboring subspaces. The probability of the ``bottleneck region'' $\mathcal{B}$ can be used to upper bound the time it takes to cross from $\mathcal{A}$ to $\mathcal{C}$.
    }
    \label{fig:fig1}
\end{figure}

{\it Quantum bottleneck theorem.} We will now formulate the quantum bottleneck theorem. Informally, it says the following. Imagine that the Hilbert space can be split into three regions (subspaces) $\mathcal{A},\mathcal{B},\mathcal{C}$ such that any quantum trajectory from $\mathcal{A}$ to $\mathcal{C}$ has to pass through~$\mathcal{B}$. Then, if the steady state $\rho_{\rm ss}$ has small weight on $\mathcal{B}$ compared to $\mathcal{A}$, then we can construct a state supported entirely on $\mathcal{A}$ that will remain localized in $\mathcal{A}$ for long times. For our proof, we need to further split the bottleneck region $\mathcal{B}$ into two parts, $\mathcal{B}_1$, $\mathcal{B}_2$, so that quantum trajectories follow the sequence $\mathcal{A} \rightarrow \mathcal{B}_1 \rightarrow \mathcal{B}_2 \rightarrow \mathcal{C}$ (see \autoref{fig:fig1}). 

\begin{thm}[Quantum bottleneck theorem]\label{thm:bottleneck_general}
    Let $\mathcal{M}$ be a quantum channel with steady state $\rho$ and let $\hilbert = \mathcal{A} \oplus \mathcal{B}_1 \oplus \mathcal{B}_2\oplus\mathcal{C}$ be a decomposition of the Hilbert space into orthogonal subspaces with projectors $P_\mathcal{A},P_{\mathcal{B}_1},P_{\mathcal{B}_2},P_{\mathcal{C}}$ and $P_{\mathcal{B}_1}+P_{\mathcal{B}_2} \equiv P_\mathcal{B}$. Assume that $\mathcal{M}$ has a Kraus representation such that all Kraus operators $K_i$ satisfy 
    \begin{align}\label{eq:bottleneck_condition_general}
        (P_\mathcal{C} + P_{\mathcal{B}_2}) K_i P_\mathcal{A} = (P_\mathcal{A}+P_{\mathcal{B}_1}) K_i P_\mathcal{C} = 0.
    \end{align}
    Then the projected state, $\rho_\mathcal{A} \equiv P_\mathcal{A} \rho P_\mathcal{A} / \tr(P_\mathcal{A}\rho)$ satisfies
    \begin{equation}\label{eq:bottleneck_ratio_general}
        ||\mathcal{M}[\rho_\mathcal{A}] - \rho_\mathcal{A}||_1 \leq 10 \frac{||P_\mathcal{B}\rho||_1}{\tr(P_\mathcal{A}\rho)} \equiv 10 \Delta.
    \end{equation}
\end{thm}
We sketch the proof of \cref{thm:bottleneck_general} in the Appendix and we provide a full proof in the Supplement~\cite{SM_bottleneck}. In the Appendix we also state a different version of the theorem that does not rely on splitting of $\mathcal{B}$ into $\mathcal{B}_1$ and $\mathcal{B}_2$ but leads to a slightly weaker bound. 

\autoref{eq:bottleneck_ratio_general} says that the rate of change of $\rho_\mathcal{A}$ is small whenever the weight of the steady state $\rho$ on $\mathcal{B}$ is sufficiently small compared to that of $\mathcal{A}$. Note, however, that while the denominator of the bottleneck ratio $\Delta$ in \cref{eq:bottleneck_ratio_general} is simply the probability associated to the subspace $\mathcal{A}$, the numerator does not have a similar interpretation because the projector $P_\mathcal{B}$ appears only on one side of $\rho$. This means that in decomposing $\rho$ into blocks according to $\mathcal{A,B,C}$, we require not only that the diagonal block $\rho_\mathcal{BB}$ is small, but also that the off-diagonal blocks $\rho_\mathcal{AB}$ and $\rho_\mathcal{BC}$ are small also (relative to $\rho_\mathcal{AA}$). However, off-diagonal and diagonal matrix elements are related by the fact that $\rho$ is positive semi-definite. Using this, in the Supplement~\cite{SM_bottleneck},  we upper bound the numerator solely in terms of probabilities using the following inequality:
\begin{equation}\label{eq:diag_bound}
    ||P_\mathcal{B} \rho||_1 \leq \sqrt{\tr(P_\mathcal{B}\rho)}.
\end{equation}
Note that this bound might be very loose. For example, if $[P_\mathcal{B},\rho] = 0$ then $||P_\mathcal{B}\rho||_1 = \tr(P_\mathcal{B}\rho)$, without the square root. This happens e.g. in classical systems. 

\autoref{eq:diag_bound} provides a more physical interpretation of the bottleneck ratio $\Delta$ in \autoref{eq:bottleneck_ratio_general}: $\rho_\mathcal{A}$ becomes an approximate steady state if the probability of the bottleneck region $\mathcal{B}$ is sufficiently smaller than the probability of $\mathcal{A}$ itself. However, the presence of the square root means that this effect is weaker than in the classical case, which stems from the possible presence of the aforementioned off-diagonal matrix elements. This is somewhat reminiscent of Grover-like quadratic quantum speedups \cite{grover1996algorithm,byrnes2018generalized}.

\cref{thm:bottleneck_general} applies to any channel satisfying its condition, without reference to Markovianity. However, for a Markovian evolution whose updates $\mathcal{M}_\tau$ all satisfy these conditions, \cref{thm:bottleneck_general} directly implies a bound on the mixing time. In particular, let $\{\mathcal{M}_\tau\}$ be a collection of channels that all share the same steady state $\rho$ and all satisfy \cref{eq:bottleneck_condition_general} with respect to the same decomposition into $\mathcal{A},\mathcal{B}_1,\mathcal{B}_2,\mathcal{C}$. A simple argument~\cite{SM_bottleneck} then shows that
\begin{align}
    ||\mathcal{M}_t \ldots\mathcal{M}_1[\rho_\mathcal{A}] - \rho_\mathcal{A}||_1 \leq 10 t \Delta,
\end{align}
which in turn implies a lower bound on the mixing time, which follows from a telescoping sum~\cite{SM_bottleneck}:
\begin{corollary}\label{lem:tmix_bound}
    $t_{\mathrm{mix},\varepsilon} \geq \frac{1 - \tr(P_\mathcal{A}\rho)  - \varepsilon}{5 \Delta}$. 
\end{corollary} 
 An important aspect of \cref{thm:bottleneck_general} and \cref{lem:tmix_bound} is that they depend only on the properties of the steady state $\rho$. The properties of the channels $\mathcal{M}_\tau$ appear only through condition~\eqref{eq:bottleneck_condition_general}. In the following, we will consider a situation where this condition follows from $k$-locality and thus applies to a large family of channels. 

{\it Bottlenecks from locality.}  Theorem~\ref{thm:bottleneck_general} leaves open the question of how to find the subspaces $\mathcal{A,B,C}$. While the appropriate choice of $\mathcal{A}$ will depend on the problem at hand, once $\mathcal{A}$ is specified, there is a natural definition of $\mathcal{B}$ (and consequently $\mathcal{C}$) through \emph{locality}, which applies to many cases of interest, as we now explain. Consider a system composed of $n$ qubits (generalization to qudits is straightforward). An operator is $r$-local if it is a linear combination of operators acting on at most $r$ qubits. For channels, we say that $\mathcal{M}$ is $r$-local if it has a Kraus representation where all Kraus operators are $r$-local.

Locality allows us to define a notion of proximity of a state in $\hilbert$ to another state, or to some subspace of $\hilbert$, which we will use to define the bottleneck subspace $\mathcal{B}$ appearing in Thm.~\ref{thm:bottleneck_general}. In particular, we define the neighborhood and the boundary of a subspace as follows (\autoref{fig:fig2}):
\begin{Def}
    Given a subspace $\mathcal{V} \subset \hilbert$ we define its $r$-\textbf{neighborhood} as
    \begin{equation}
        \mathcal{B}_r(\mathcal{V}) := \text{Span}\{\hat{O} \ket{\psi}| \ket{\psi} \in \mathcal{V}, \hat{O} \text{ is } r\text{-local}\},
    \end{equation}
    and its $r$-\textbf{boundary} $\partial_r \mathcal{V}$ as the orthogonal complement of $\mathcal{V}$ within $\mathcal{B}_r(\mathcal{V})$, satisfying $\mathcal{B}_r(\mathcal{V}) = \mathcal{V} \oplus \partial_r\mathcal{V}$.
\end{Def}

With these definitions, for \emph{any} $r$-local channel, a quantum trajectory initialized in $\mathcal{V}$ must pass through the boundary $\partial_r \mathcal{V}$ on its way to the rest of Hilbert space. This means that we can apply \cref{thm:bottleneck_general} to this situation to obtain a bound applicable to any local channel~\cite{SM_bottleneck}:

\begin{thm}[Quantum Bottleneck Theorem, local channels]\label{thm:bottleneck_local}
    Let $\rho$ be a density matrix, $\mathcal{V}\subset\hilbert$ a subspace with orthogonal projector $P_\mathcal{V}$, and $\rho_\mathcal{V} = P_\mathcal{V}\rho P_\mathcal{V} / \tr(P_\mathcal{V} \rho)$. We define the bottleneck ratio as
    \begin{equation}\label{eq:bottleneck_ratio_local}
        \Delta(r) := \frac{||P_{\partial_{2r}\mathcal{V}}\rho||_1}{\tr(P_{\mathcal{V}}\rho)}.
    \end{equation}
    Let $\mathcal{M}$ be a quantum channel satisfying $\mathcal{M}[\rho] = \rho$ and assume that $\mathcal{M}$ is an $r$-local channel. Then
\begin{equation}\label{eq:bottleneck_thm_local}
        ||\mathcal{M}[\rho_\mathcal{V}] - \rho_\mathcal{V}||_1 \leq 10 \Delta(r).
    \end{equation}
\end{thm}

To see the implications of this theorem, consider a family of problems with different sizes $n$ that have some well-defined thermodynamic limit $n\to\infty$. Assume that we can find a family of subspaces $\mathcal{V}(n)$ and a diverging function $r(n)$ such that the bottleneck ratio $\Delta(n) \equiv \Delta(r(n))$ defined in \cref{eq:bottleneck_ratio_local} goes to zero superlogarithmically in the thermodynamic limit. Then \cref{eq:bottleneck_thm_local} implies that for any $s$-local channel $\mathcal{M}$ with finite $s$ (or, more generally, $s = o(r(n))$), the projected state $\rho_\mathcal{V}$ becomes an increasingly good approximate steady state of $\mathcal{M}$. In particular, assuming that $\lim_{n\to\infty}\tr(P_{\mathcal{V}(n)}\rho_n) < 1$, the mixing time is lower bounded to be $t_\text{mix} = \Omega(1/\Delta(n))$. 

\cref{thm:bottleneck_local} generalizes to \emph{quasi-local} channels, i.e. those that are well-approximated by local ones. In particular, if $f(r)$ denotes the error we incur (in diamond norm) when we approximate $\mathcal{M}$ with an $r$-local channel, the upper bound in Eq.~\eqref{eq:bottleneck_thm_local} becomes $10\Delta(r) + f(r)$~\cite{SM_bottleneck}. Hence, as long as $f(r)$ decays sufficiently fast (e.g. superlogarithmically) with $r$, we still get a meaningful lower bound on the mixing time. In particular, this applies to recently constructed quasi-local Gibbs samplers.

\begin{figure}
    \centering
    \includegraphics[width=\columnwidth]{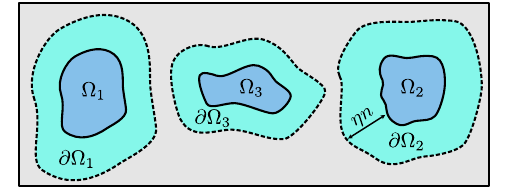}
    \caption{{\bf Local channels and Gibbs state structure.} For local channels, we define bottleneck regions in terms of states that can be reached by local operations. In non-trivial finite temperature phases, the Gibbs state $\rho_\beta$ fragments into multiple components surrounded by large bottlenecks.
    }
    \label{fig:fig2}
\end{figure}

{\it Gibbs sampling and phases of matter.}
The bottleneck theorem has important implications for the the problem of \emph{Gibbs sampling}, i.e. the case where the steady state is $\rho = \rho_{\beta} \propto e^{-\beta H}$ for some (usually local) Hamiltonian $H$. The construction of general (quasi-)local Gibbs samplers is a very recent development~\cite{chen2311efficient,gilyen2024quantum,ding2024efficient}. It has been shown that for local $H$ at sufficiently small $\beta$ they achieve \emph{rapid mixing}, $t_{\rm mix} = O(\log{n})$~\cite{rouze2024optimal}. Our results complement this by establishing sufficient conditions for \emph{slow mixing}. In particular, we expect that \cref{thm:bottleneck_local} implies (super)polynomial mixing times at low temperatures in many cases of interest for reasons we now explain. 

For Gibbs states, the bottleneck ratio is related to the notion of a \emph{free energy barrier}, between $\mathcal{V}$ and the rest of Hilbert space. Let us define the free energy of subspace $\mathcal{V}$ as $F(\mathcal{V}) = -\tr(e^{-\beta H} P_\mathcal{V}) / \beta$. We can relate the bottleneck ratio~\cref{eq:bottleneck_ratio_local} to the free energy in various cases (in the following, we omit the $r$-dependence for simplicity):
\begin{itemize}
    \item[(a)] If $\mathcal{V}$ has finite probability in the thermodynamic limit, $\tr(P_\mathcal{V}\rho_\beta) > c$, then $\Delta \leq e^{-\beta F(\partial\mathcal{V})/2} / c$.
    \item[(b)] If $[\rho_\beta,P_{\partial\mathcal{V}}] = 0$, then $\Delta \leq e^{-\beta(F(\partial\mathcal{V})-F(\mathcal{V}))}$.
    \item[(c)] In the most general case, from \autoref{eq:diag_bound}, 
    \begin{equation}\label{eq:thermo_bound}
        \Delta \leq e^{-\frac{\beta}{2}((F(\partial\mathcal{V})-F(\mathcal{V})) + (F - E_{\text{min}}(\mathcal{V})))},
    \end{equation}
    where $F = F(\hilbert)$ is the free energy of the Gibbs state and $E_{\text{min}}(\mathcal{V}) \equiv \min_{\ket{\psi}\in\mathcal{V}}(\langle \psi | H | \psi \rangle)$ is the minimal energy in the subspace $\mathcal{V}$. As long as this minimal energy is smaller than $F$, this reduces to the previous bound with $\beta\to\beta/2$.
\end{itemize}
The upshot is that in many cases, either $F(\partial\mathcal{V})$ or $F(\partial\mathcal{V}) - F(\mathcal{V})$ diverging with $n$ (a ``macroscopic free energy barrier'') implies $\Delta(n)\to 0$. Physically, one expects such macroscopic barriers to exist in many cases when the system realizes some non-trivial phase of matter, e.g. between different symmetry broken ordered states. We will now discuss some examples of this.

Proofs of macroscopic free energy barriers exist in a variety of \emph{classical} spin Hamiltonians, such as the $D\geq 2$ dimensional Ising model~\cite{levin2017markov,thomas1989bound} and certain classical spin glasses~\cite{mezard2009information,arous2024shattering,alaoui2023shattering,gamarnik2023shattering}. Our proof thus shows that the same slow mixing properties that exist for classical Markov chain Monte Carlo methods for these models also apply to any \emph{quantum} Gibbs sampler. Similar arguments can be used to establish macroscopic free-energy barriers in certain commuting quantum Hamiltonians, such as the four dimensional toric code~\cite{dennis2002topological,hastings2011topological,alicki2010thermal}. Similarly, certain spatially non-local quantum error correcting codes can host an exponentially large number of stable subspaces, similar to classical spin glasses, as we have shown recently~\cite{placke2024tqsg}. 

These examples all fall under (b) above, corresponding to commuting Hamiltonians. However, on physical grounds, one expects that the properties of the Gibbs state in these cases are stable to various weak perturbations, which violate this commuting condition. A great strength of our general quantum bottleneck theorem is that it can also be applied to these perturbed models and below we use \cref{eq:thermo_bound} to establish the resilience of free energy barriers to perturbations for a family of models in the context of classical and quantum spin glasses (see also Ref. \cite{gamarnik2024slow} for a related discussion concerning perturbations of the 2D Ising model)\footnote{Another potential application is spontaneous breaking of \emph{continuous} symmetries, where one expects macroscopic free energy barriers but no commuting Hamiltonian description exists.}.

This suggests a unifying picture of many non-trivial finite temperature phases of matter (see \autoref{fig:fig2}). For these, the Gibbs state fragments into multiple components that are separated by bottlenecks. Different phases are distinguished by the number and nature of these components, allowing to treat many different phases within the same unified framework. This way, the bottleneck theorem connects the study of finite temperature equilibrium phases to that of phases in open quantum systems~\cite{rakovszky2024defining}, by suggesting that the former can be understood in terms of the near-degenerate steady states that must be present for any quasi-local Gibbs sampler. 

{\it Bottlenecks from extensive energy barriers.} Here, we discuss a family of concrete examples where \cref{thm:bottleneck_local} leads to exponential mixing times, $t_{\rm mix} = e^{\Omega(n)}$. 

As we discussed, the bottleneck ratio simplifies when $\rho_\beta$ and $P_{\partial\mathcal{V}}$ commute. A situation where this applies is the case of Hamiltonians $H_0$, which can be written as sums of commuting projectors. One often imposes a ``low density'' condition, meaning that each term acts on finitely many qubits and each qubit appears in finitely many terms (with ``finitely'' meaning $O(1)$ in the limit $n\to\infty$). While this condition implies a degree of locality, it nevertheless allows for $H_0$ that are not geometrically local. An interesting family of such stabilizer models are those exhibiting \emph{extensive energy barriers}~\cite{sipser_spielman1996,dinur2023qldpc,hong2024quantum}:
\begin{equation}\label{eq:extensive_energy}
E_{\text{min}}(\partial\mathcal{V}) - E_{\text{min}}(\mathcal{V}) \geq \kappa n,
\end{equation}
for some constant $\kappa > 0$. It is easy to see that this also leads to a bottleneck ratio vanishing as $\Delta(n) = e^{-\Omega(n)}$. In Refs. \cite{placke2025expansion,placke2024tqsg}, it is shown that in some models, there exist exponentially many in $n$ subspaces $\mathcal{V}$ satisfying \cref{eq:extensive_energy}, including ones with $E_{\text{min}}(\mathcal{V}) = O(n)$, which is interpreted in terms of spin glass order~\cite{mezard1987spin,mezard2009information} in these models. 

A natural question to ask is about the robustness of this spin glass order to perturbations of the Hamiltonian $H = H_0 +V$ where $V$ is a sum of $r$-local terms. Here we use \cref{thm:bottleneck_local} to establish stability, proving~\cite{SM_bottleneck}

\begin{thm}[Stability of spin glass order]\label{thm:stability}
    Let $\mathcal{V}$ be a subspace satisfying \cref{eq:extensive_energy} for a commuting projector Hamiltonian $H_0$, with $E_{\text{min}}(\mathcal{V}) \leq \kappa n /2$. Consider $H = H_0 + V$ with $||V|| \leq g n$. For $\beta > \beta_*$ and $g < g_*(\beta)$, the bottleneck ratio of $\Delta(\mathcal{V})$ in the state $\rho_\beta \propto e^{-\beta H}$ remains exponentially small in system size, $\Delta(\mathcal{V}) = e^{-\Omega(n)}$. This implies that the mixing time is $t_{\rm mix} = \Omega(n)$.
\end{thm}

The key step in proving \cref{thm:stability}, beyond applying \cref{eq:extensive_energy}, is to show that eigenstates of $H$ with energies close to $E_{\rm min}(\mathcal{V})$ have exponentially small amplitudes in $\partial\mathcal{V}$. We establish this using recent results from Ref. \cite{yin2024eigenstate}.

{\it Outlook.} Our work opens up a number of interesting open questions. One is to use the bottleneck theorem to prove large mixing times in models of interest, such as those in the four dimensional toric code phase. Another is to explore the implications of this dynamical perspective on thermal phases, e.g. by applying the ideas developed in Ref.\ \cite{rakovszky2024defining} to Gibbs samplers. Finally, we expect our bottleneck result to have applications for other Monte Carlo-style quantum algorithms beyond Gibbs sampling. 

{\it Note added.} During the completion of this manuscript, a related work appeared, also developing quantum versions of the bottleneck theorem~\cite{gamarnik2024slow}. While their original version applies only to the case when the decomposition of the Hilbert space commutes with the steady state, and is thus less broadly applicable than ours, the updated version, appearing in the same arXiv posting as the present manuscript, involves a bottleneck theorem that should apply to similar cases as ours. The definition of the quantum bottleneck ratio is nevertheless different in the two cases and a more detailed comparison is an interesting problem for future work. 

{\it Acknowledgements.} 
We thank Grace Sommers for collaboration on related work and we thank Ethan Lake, Andr\'as Gily\'en, Nicholas O'Dea and especially Curt von Keyserlingk for useful discussions.
T.R. was supported in part by Stanford Q-FARM Bloch Postdoctoral Fellowship in Quantum Science and Engineering, by the HUN-REN Welcome Home and Foreign Researcher Recruitment Programme 2023 and by the Supported Research Groups Programme, HUN-REN-BME-BCE Quantum Technology Research Group (TKCS-2024/34).
B.P. acknowledges funding through a Leverhulme-Peierls Fellowship at the University of Oxford and the Alexander von Humboldt foundation through a Feodor-Lynen fellowship.
This work was done in part while N.P.B. was visiting the Simons Institute for the Theory of Computing, supported by DOE QSA grant \#FP00010905.
V.K. acknowledges support from the Packard Foundation through a Packard
Fellowship in Science and Engineering and the Office of Naval Research Young Investigator Program (ONR YIP) under Award Number N000142412098.

\nocite{rakovszky2023physics,leverrier2015quantum_expander}
\bibliography{references.bib}

\clearpage

\appendix

\renewcommand{\thesubsection}{\Alph{subsection}}

\begin{center}

 {\large \textbf{End Matter}}

\end{center}

\subsubsection{Appendix A: Alternative bottleneck theorem}. The version of the bottleneck theorem in \cref{thm:bottleneck_general} relies on splitting the bottleneck subspace $\mathcal{B}$ into two parts, $\mathcal{B}_1$ and $\mathcal{B}_2$. This leads to a fairly straightforward proof (see \cite{SM_bottleneck}) and this version of the theorem applies straightforwardly to the physical examples we focused on, where the bottleneck arises from locality. In that case, the splitting merely amounted to having to make the ``width'' $r$ or the bottleneck twice as large (see \cref{thm:bottleneck_local}). 

However, in other contexts, finding an appropriate splitting of the bottleneck subspace might pose more difficulty and it would be beneficial to have a version of the theorem that does not rely on this. We we state such a theorem (proved in the Supplement \cite{SM_bottleneck}). As we will see, the downside is that it leads to a weaker bound compared to \cref{thm:bottleneck_general}.

\begin{thm}[Quantum Bottleneck Theorem, alternative version]\label{thm:bottleneck_general_v2}
    Let $\mathcal{M}$ be a quantum channel with steady state $\rho$. Let $\hilbert = \mathcal{A} \oplus \mathcal{B} \oplus \mathcal{C}$ be a decomposition of the Hilbert space into three orthogonal subspaces and let $P_\mathcal{A},P_{\mathcal{B}},P_{\mathcal{C}}$ denote the corresponding projectors. Assume that $\mathcal{M}$ has a Kraus representation such that all Kraus operators $K_i$ satisfy 
    \begin{align}\label{eq:bottleneck_condition_v2}
         P_\mathcal{C} K_i P_\mathcal{A} = P_\mathcal{A} K_i P_\mathcal{C} = 0.
    \end{align}
    Then the projected state, $\rho_A \equiv P_\mathcal{A} \rho P_\mathcal{A} / \tr(P_\mathcal{A}\rho)$ satisfies
    \begin{equation}\label{eq:bottleneck_ratio_general_v2}
        ||\mathcal{M}[\rho_\mathcal{A}] - \rho_\mathcal{A}||_1 \leq c \frac{\sqrt{||P_{\mathcal{B}}\rho||_1}}{\tr(P_\mathcal{A}\rho)},
    \end{equation}
    where $c = 2(\sqrt{3}+1) + 3(\sqrt{5}+1)$ is a constant.
\end{thm}
Thus, \cref{thm:bottleneck_general_v2} is more general, involing a decomposition of the Hilbert space into three, not four, orthogonal subspaces. The price to be paid for this is that the numerator in the bottleneck ratio~\autoref{eq:bottleneck_ratio_general_v2} the numerator $||P_\mathcal{B}\rho||_1$ is replaced by its square root, leading to a weaker bound. We do not know whether this is a consequence of our proof technique (described below) or is in fact the best possible bound possible given the assumptions. In the main text, we focused on the version in \cref{eq:bottleneck_ratio_general} since it applies straightforwardly to the physical examples we considered, where bottlenecks are obtained from locality.

\subsubsection{Appendix B: Proof sketch for \cref{thm:bottleneck_general} and \cref{thm:bottleneck_general_v2}}

The proofs of \cref{thm:bottleneck_general} and \cref{thm:bottleneck_general_v2} both stat with the same initial steps, but \cref{thm:bottleneck_general_v2} requires proving an additional lemma while \cref{thm:bottleneck_general} admits a more straightforward proof due to its stronger assumptions. Here, we present a sketch of both proofs combined, leaving the full details to the Supplement. 

Before going into the proofs, let us first describe the intuition behind them. The basic idea is to first notice that slowness of the state projected onto $\mathcal{A}$ is equivalent to slowness of the state projected onto the complement $\mathcal{BC}$. Considering this latter state and asking how it changes upon applying the dynamics $\mathcal{M}$, and using the fact that there are no direct transitions from $\mathcal{C}$ to $\mathcal{A}$, one finds that the only change in the projected state occurs due to a flow of probability from $\mathcal{B}$ to $\mathcal{A}$. However, the size of this change is bounded by the amount of probability contained in $\mathcal{B}$ initially, which is what is being captured by the (numerator of) the bottleneck ratio. In fact, for classical Markovian dynamics, this argument can be easily turned into a rigorous proof, which yields that the rate of change of the projected state is upper bounded by the ratio of probabilities associated to $\mathcal{B}$ and $\mathcal{A}$. However, in the quantum case, the existence of off-diagonal elements yields complications that have to be dealt with either by weakening the bound or by strengthening the assumptions as in \cref{thm:bottleneck_general}.

Let us now turn to how this is done in more detail. We first write $P_\mathcal{A} = 1\!\!1 - P_\mathcal{BC}$. Then, using $\mathcal{M}[\rho] = \rho$, bounding $||\mathcal{M}[P_\mathcal{A}\rho P_\mathcal{A}] - P_\mathcal{A}\rho P_\mathcal{A}||_1$ becomes equivalent to bounding the following three expressions:
    \begin{align}\label{eq:3terms}
        &||\mathcal{M}[P_\mathcal{BC}\rho] - P_\mathcal{BC}\rho||_1, \nonumber \\ &||\mathcal{M}[ \rho P_\mathcal{BC}] -  \rho P_\mathcal{BC}||_1, \nonumber \\
        &||\mathcal{M}[P_\mathcal{BC} \rho P_\mathcal{BC}] - P_\mathcal{BC} \rho P_\mathcal{BC}||_1.
    \end{align}
In the following we will focus on the first of these; the other two can be bounded in a similar manner (see~\cite{SM_bottleneck}).

In proving \cref{thm:bottleneck_general}, one inserts an identity in the form $1\!\!1 = (P_\mathcal{A}+P_{\mathcal{B}_1}) + (P_{\mathcal{B}_2} + P_\mathcal{C}) \equiv P_{\mathcal{A}\mathcal{B}_1} + P_{\mathcal{B}_2\mathcal{C}}$ and uses~\autoref{eq:bottleneck_condition_general} to write 
\begin{align}\label{eq:proof_1st_term}
    \mathcal{M}[P_\mathcal{BC}\rho] = P_{\mathcal{A}\mathcal{B}_1} \mathcal{M}[P_\mathcal{B}\rho] + P_{\mathcal{B}_2 \mathcal{C}} \rho.
\end{align}
Thus, crucially, due to~\autoref{eq:bottleneck_condition_general}, in the first term we could replace $P_\mathcal{BC}\rho$ with $P_\mathcal{B}\rho$ in the argument. We can then use properties of the trace norm to bound $||P_{\mathcal{A}\mathcal{B}_1} \mathcal{M}[P_\mathcal{B}\rho]||_1 \leq ||P_\mathcal{B}\rho||_1$ as we wanted. Comparing the last term in~\autoref{eq:proof_1st_term} with the first line of~\autoref{eq:3terms} we find $P_{\mathcal{B}_2 \mathcal{C}} \rho - P_\mathcal{BC}\rho = P_{\mathcal{B}_1} \rho$ which has $||P_{\mathcal{B}_1} \rho||_1 \leq ||P_\mathcal{B}\rho||_1$ as well. Combining all terms also from the second and third lines of~\autoref{eq:3terms} gives the bound in \cref{thm:bottleneck_general}.

To prove \cref{thm:bottleneck_general_v2}, we have to work harder. We now no longer have the luxury of splitting $\mathcal{B}$ into two parts. Instead, we need to insert the identity in the form $1\!\!1 = P_\mathcal{AB} + P_\mathcal{C}$, which gives, using~\autoref{eq:bottleneck_condition_v2}, $\mathcal{M}[P_\mathcal{BC}\rho] = P_\mathcal{AB} \mathcal{M}[P_\mathcal{BC}\rho] + P_\mathcal{C}\rho$. The second term, $P_\mathcal{C}\rho$ will again cancel with part of $P_\mathcal{BC}\rho$ from the first line of~\autoref{eq:3terms}, leaving behind $P_\mathcal{B}\rho$ as we wanted. However, we still have to deal with the first term, $P_\mathcal{AB} \mathcal{M}[P_\mathcal{BC}\rho]$, which now still involves $P_\mathcal{BC}\rho$, rather than $P_\mathcal{B}\rho$, in its argument. 

The proceed further, let us first make the following observation. Imagine that the state $\mathcal{M}[P_\mathcal{BC}\rho]$ was block-diagonal with respect to the decomposition of the Hilbert space into $\mathcal{H} = \mathcal{AB} \oplus \mathcal{C}$ (combining $\mathcal{A}$ and $\mathcal{B}$ into a single subspace). If this were so (which includes the case of classical Markovian dynamics), then we would have that 
\begin{align}\label{eq:norm_addition}
        ||\mathcal{M}[P_\mathcal{BC}\rho]||_1 &= ||P_\mathcal{AB} \mathcal{M}[P_\mathcal{BC}\rho] + P_\mathcal{C}\rho||_1 \nonumber \\ &= ||P_\mathcal{AB} \mathcal{M}[P_\mathcal{BC}\rho]||_1 + ||P_\mathcal{C}\rho||_1,
\end{align}
i.e., the norms would simply be additive. On the other hand, quantum channels do not increase the trace norm and therefore
\begin{align}\label{eq:norm_decrease}
    ||\mathcal{M}[P_\mathcal{BC}\rho]||_1 \leq ||P_\mathcal{BC}\rho||_1 \leq ||P_\mathcal{B}\rho||_1 + ||P_\mathcal{C}\rho||_1.
\end{align}
Comparing~\autoref{eq:norm_addition} with~\autoref{eq:norm_decrease} we find that 
\begin{equation}\label{eq:bound_blockdiag}
    ||P_\mathcal{AB} \mathcal{M}[P_\mathcal{BC}\rho]||_1 \leq ||P_\mathcal{B}\rho||_1.    
\end{equation} 
This captures the intuition mentioned at the beginning of this appendix: the flow of probability from $\mathcal{BC}$ to $\mathcal{A}$ is limited by the probability initially contained in $\mathcal{B}$. This would again yield that the first line of~\autoref{eq:3terms} is bounded from above by $2||P_\mathcal{B}\rho||_1$, but now only using the weaker assumption of~\autoref{eq:bottleneck_condition_v2}.

This conclusion, however, was predicated on the assumption that $\mathcal{M}[P_\mathcal{BC}\rho]$ is block-diagonal with respect to $\mathcal{H} = \mathcal{AB} \oplus \mathcal{C}$. This need not be true in general, in which case~\autoref{eq:norm_addition}, and consequently~\autoref{eq:bound_blockdiag}, fails. Is there a more general version of~\autoref{eq:norm_addition}? In the Supplement we show that indeed there is, proving that 
\begin{align}\label{eq:norm_subaddition}
        ||\mathcal{M}[P_\mathcal{BC}\rho]||_1 &= ||P_\mathcal{AB} \mathcal{M}[P_\mathcal{BC}\rho] + P_\mathcal{C}\rho||_1 \nonumber \\ &\geq \sqrt{||P_\mathcal{AB} \mathcal{M}[P_\mathcal{BC}\rho]||_1^2 + ||P_\mathcal{C}\rho||_1^2}.
\end{align}
We can combine this weaker bound with~\autoref{eq:norm_decrease} (which always holds) to get a more general version of~\autoref{eq:bound_blockdiag} which says that
\begin{equation}\label{eq:bound_general}
    ||P_\mathcal{AB} \mathcal{M}[P_\mathcal{BC}\rho]||_1 \leq \sqrt{||P_\mathcal{B}\rho||_1(2 + ||P_\mathcal{B}\rho||_1)}.
\end{equation}
Importantly, by replacing~\autoref{eq:norm_addition} (which was specific to the block-diagonal case) with~\autoref{eq:norm_subaddition} (which is general), we ended up with a weaker bound that has $\sqrt{||P_\mathcal{B}\rho||_1}$, rather than $||P_\mathcal{B}\rho||_1$, on the right hand side (compare~\autoref{eq:bound_general} with~\autoref{eq:bound_blockdiag}).

Inserting the bound~\autoref{eq:bound_general} back into the first line of~\autoref{eq:3terms} one finds
\begin{align}
    ||\mathcal{M}[P_\mathcal{BC}\rho] - P_\mathcal{BC}\rho||_1 &\leq ||P_\mathcal{AB} \mathcal{M}[P_\mathcal{BC}\rho]||_1 + ||P_\mathcal{B}\rho||_1 \nonumber \\
   &\leq (\sqrt{3}+1) \sqrt{||P_\mathcal{B}\rho||_1}.
\end{align}
One gets the same bound for the second line of~\autoref{eq:3terms} and a similar one for the third line (see~\cite{SM_bottleneck}). Combining all of them yields \cref{thm:bottleneck_general_v2}.
 
\end{document}

\title{Supplementary Material for:\\
Bottlenecks in quantum channels and finite temperature phases of matter
}

\maketitle

\section{A classical bottleneck theorem}

Here we present a classical analog of our Theorem 1, applicable to classical Markov chains. While bottleneck theorems in this case are well known~\cite{levin2017markov}, our formulation differs from the standard one so we include it here to showcase the logic of the more general quantum proof in a simpler setting. 

Let $\Lambda$ denote the state space. Probability distributions $\ket{\pi}$ over $\Lambda$ correspond to $|\Lambda|$-dimensional non-negative vectors with $||\ket{\pi}||_1 = \langle \Lambda |\pi\rangle = 1$, where we use $\ket{\Lambda}$ to denote the all-$1$ vector and $||\cdot||_1$ is the vector 1-norm (sum of absolute values of all elements). We will use $P_\Omega$ to denote the projector onto configurations in the set $\Omega \subseteq \Lambda$. The total probability weight associated to this set is $\pi(\Omega) = \langle \Lambda | P_\Omega |\pi\rangle \equiv \langle \Omega|\pi\rangle$ where $\ket{\Omega} \equiv P_\Omega \ket{\Lambda}$. 

We are now ready to state the analog of the bottleneck theorem.
\begin{theorem}[Classical Bottleneck Theorem]\label{thm:classical_bottleneck_general}
    Let $M$ be a Markov generator (stochastic matrix) on $\Lambda$, with steady state $\ket{\pi}$, and let $\Lambda = A \uplus B_1 \uplus B_2 \uplus C$ be a partition of the state space into a union of four disjoint sets, such that the following condition is satisfied:
    \begin{equation}\label{eq:classical_cond}
    P_{B_2\uplus C}M P_A = P_{A\uplus B_1} M P_C = 0    
    \end{equation}
    Define the projected state $\ket{\pi_A} \equiv P_A \ket{\pi} / \pi(A)$. Then
    \begin{equation}
        ||M\ket{\pi_A} - \ket{\pi_A}||_1 \leq 2\frac{ \pi(B)}{\pi(A)},
    \end{equation}
    where $B \equiv B_1 \uplus B_2$.
\end{theorem}

\begin{proof}
    First, let us introduce some notation. We will denote by $P_{AB_1} = P_{A\uplus B_1} = P_A + P_{B_1}$ and similarly for other combinations of the four subsets. We thus have $P_{AB_1B_2C} = P_{ABC} = \id$. 
        
    The conditions on $M$ mean that in order to go from a configuration in $A$ to a configuration in $C$, the system has to first go through $B_1$ and then $B_2$, in that order, so that the sets are naturally arranged in a sequence (see Fig. 1 of the main text). 

    The claim of the theorem is equivalent to 
    \begin{equation}
        ||MP_{A} \ket{\pi} - P_{A}\ket{\pi}||_1 \leq 2  || P_{B} | \pi \rangle||_1.
    \end{equation}
    To show this, we first note that 
    \begin{equation}
        M P_A \ket{\pi} - P_A \ket{\pi} = M(1\!\!1-P_{BC}) \ket{\pi} - P_A \ket{\pi} = \ket{\pi} - P_A \ket{\pi} - M P_{BC}\ket{\pi} = -(MP_{BC}\ket{\pi} - P_{BC}\ket{\pi}).
    \end{equation}
    We now multiply $MP_{BC}\ket{\pi}$ by $\id = P_{AB_1} + P_{B_2C}$ from  the left and note that, due to the conditions Eq.~\eqref{eq:classical_cond}, the following two equalities hold: 
    \begin{align} 
        P_{B_2C} M P_{BC} \ket{\pi} &= P_{B_2C} M \ket{\pi} = P_{B_2C} \ket{\pi}, \nonumber \\
        P_{AB_1} M P_{BC} \ket{\pi} &= P_{AB_1} M P_{B} \ket{\pi}.
    \end{align}
    Importantly, the first term cancels part of $P_{BC}\ket{\pi} = (P_{B_1} + P_{B_2C})\ket{\pi}$. Combining these, we thus get 
        \begin{equation}
        ||MP_{BC} \ket{\pi} - P_{BC}\ket{\pi}||_1 = || P_{AB_1} M P_{B} \ket{\pi} - P_{B_1}\ket{\pi}||_1 \leq ||P_{AB_1} M P_{B} \ket{\pi}||_1 + ||P_{B}\ket{\pi}||_1 \leq 2|| P_{B} \ket{\pi} ||_1,
    \end{equation}
    where we used that both $M$ and multiplying py projections cannot increase the $1$-norm.
\end{proof}

\section{Quantum bottleneck theorems}

We consider quantum channels $\mathcal{M}$ acting on a quantum system with Hilbert space $\hilbert$. We denote by $\mathcal{V} \subseteq \mathcal{H}$ a linear subspace of $\hilbert$ and by $P_\mathcal{V}$ the orthogonal projector onto $\mathcal{V}$. 

\subsection{General theorem}

Let us now prove the quantum bottleneck theorem stated in the main text.

\begin{theorem}[Quantum Bottleneck Theorem, general case]\label{thm:quantum_bottleneck_general}
    Let $\mathcal{M}$ be a quantum channel with steady state $\rho$. Let $\hilbert = \mathcal{A} \oplus \mathcal{B}_1 \oplus \mathcal{B}_2 \oplus \mathcal{C}$ be a decomposition of the Hilbert space into four orthogonal subspaces and let $P_\mathcal{A},P_{\mathcal{B}_1},P_{\mathcal{B}_2},P_{\mathcal{C}}$ denote the corresponding projectors. Assume that $\mathcal{M}$ has a Kraus representation such that all Kraus operators $K_i$ satisfy 
    \begin{align}\label{eq:locality_condition}
        (P_{\mathcal{B}_2} + P_\mathcal{C}) K_i P_\mathcal{A} = (P_\mathcal{A}+P_{\mathcal{B}_1}) K_i P_\mathcal{C} = 0.
    \end{align}
    Then the projected state, $\rho_A \equiv P_\mathcal{A} \rho P_\mathcal{A} / \tr(P_\mathcal{A}\rho)$ satisfies
    \begin{equation}
        ||\mathcal{M}[\rho_\mathcal{A}] - \rho_\mathcal{A}||_1 \leq 10 \frac{||P_{\mathcal{B}}\rho||_1}{\tr(P_\mathcal{A}\rho)},
    \end{equation}
    where $P_{\mathcal{B}} = P_{\mathcal{B}_1} + P_{\mathcal{B}_2}$.
\end{theorem}

\begin{proof}
    First, we again introduce some shorthand notation: $P_{\mathcal{A}\mathcal{B}_1} = P_{\mathcal{A}}+P_{\mathcal{B}_1}$, $P_{\mathcal{A}\mathcal{B}} = P_{\mathcal{A}} + P_{\mathcal{B}_1} + P_{\mathcal{B}_2}$ etc. We begin by noting that, since $P_\mathcal{A} = \id - P_{\mathcal{B}\mathcal{C}}$ and $\mathcal{M}[\rho] = \rho$, \begin{equation}\label{eq:q_bottleneck_inverted}
        \mathcal{M}[P_\mathcal{A} \rho P_\mathcal{A}] - P_\mathcal{A} \rho P_\mathcal{A} = \mathcal{M}[P_{\mathcal{B}\mathcal{C}} \rho P_{\mathcal{B}\mathcal{C}} - P_{\mathcal{B}\mathcal{C}} \rho - \rho P_{\mathcal{B}\mathcal{C}}] - P_{\mathcal{B}\mathcal{C}} \rho P_{\mathcal{B}\mathcal{C}} + P_{\mathcal{B}\mathcal{C}} \rho + \rho P_{\mathcal{B}\mathcal{C}}.
    \end{equation}
    
    Consider the Kraus representation, $\mathcal{M}[\rho] = \sum_i K_i \rho K_i^\dagger$. From the conditions imposed on $K_i$ we have that
    \begin{equation}
        \mathcal{M}[P_{\mathcal{B}\mathcal{C}}\rho] = (P_{\mathcal{A}\mathcal{B}_1} + P_{\mathcal{B}_2\mathcal{C}}) \mathcal{M}[P_{\mathcal{B}\mathcal{C}}\rho] = P_{\mathcal{A}\mathcal{B}_1} \mathcal{M}[P_{\mathcal{B}}\rho] + P_{\mathcal{B}_2\mathcal{C}}\mathcal{M}[P_{\mathcal{A}\mathcal{B}\mathcal{C}} \rho] = P_{\mathcal{A}\mathcal{B}_1} \mathcal{M}[P_{\mathcal{B}}\rho] + P_{\mathcal{B}_2\mathcal{C}}\rho.
    \end{equation}
    The second of these terms will cancel with the second term of $P_{\mathcal{B}\mathcal{C}}\rho = (P_{\mathcal{B}_1} + P_{\mathcal{B}_2\mathcal{C}})\rho$. Similarly, we have
    \begin{equation}
        \mathcal{M}[\rho P_{\mathcal{B}\mathcal{C}}] = \mathcal{M}[\rho P_{\mathcal{B}}]P_{\mathcal{A}\mathcal{B}_1} + \rho P_{\mathcal{B}_2\mathcal{C}},
    \end{equation}
    where the second term will now cancel the second term of $\rho P_{\mathcal{B}\mathcal{C}} = \rho P_{\mathcal{B}_1} + \rho P_{\mathcal{B}_2\mathcal{C}}$. Finally, we have that 
    \begin{align}
        \mathcal{M}[P_{\mathcal{B}\mathcal{C}}\rho P_{\mathcal{B}\mathcal{C}}] = P_{\mathcal{B}_2\mathcal{C}} \rho P_{\mathcal{B}_2\mathcal{C}} + P_{\mathcal{A}\mathcal{B}_1} \mathcal{M}[P_{\mathcal{B}} \rho P_{\mathcal{B}}] P_{\mathcal{A}\mathcal{B}_1} + P_{\mathcal{A}\mathcal{B}_1} \mathcal{M}[P_{\mathcal{B}} \rho] P_{\mathcal{B}_2\mathcal{C}} + P_{\mathcal{B}_2\mathcal{C}} \mathcal{M}[\rho P_{\mathcal{B}}] P_{\mathcal{A}\mathcal{B}_1},
    \end{align}
    which should be compared with 
    \begin{equation}
        P_{\mathcal{B}\mathcal{C}} \rho P_{\mathcal{B}\mathcal{C}} = P_{\mathcal{B}_2\mathcal{C}}\rho P_{\mathcal{B}_2\mathcal{C}} + P_{\mathcal{B}_1} \rho P_{\mathcal{B}_1} + P_{\mathcal{B}_2\mathcal{C}} \rho P_{\mathcal{B}_1} + P_{\mathcal{B}_1} \rho P_{\mathcal{B}_2\mathcal{C}}.
    \end{equation}
    Thus, combining all of the above, we find that Eq.~\eqref{eq:q_bottleneck_inverted} becomes a sum of 16 terms, 6 of which cancel in pairs. The remaining terms all involve the projector $P_{\mathcal{B}}$ multiplying $\rho$ from either the left or the right (note that $P_{\mathcal{B}_1} \rho = P_{\mathcal{B}_1} P_{\mathcal{B}} \rho$) and we can therefore upper bound them with $||P_{\mathcal{B}}\rho||_1$. This follows from the fact that both the channel $\mathcal{M}$ and multiplication by a projector have the property that they do not increase the trace norm (the latter statement follows from H\"older's inequality: $||PA||_1 \leq ||A||_1 ||P|| \leq ||A||_1$, where $||\ldots||$ is the operator norm). For example we have that
    \begin{equation}
        ||P_{\mathcal{A}\mathcal{B}_1} \mathcal{M}[P_{\mathcal{B}}\rho]P_{\mathcal{B}_2\mathcal{C}}||_1 \leq ||\mathcal{M}[P_{\mathcal{B}}\rho]||_1 \leq ||P_{\mathcal{B}}\rho||_1.
    \end{equation}
    Thus, adding together all the remaining 10 terms, the theorem follows.
\end{proof}

If $[P_\mathcal{B},\rho] = 0$ then $||P_\mathcal{B}\rho||_1 = ||P_\mathcal{B}\rho P_\mathcal{B}||_1 = \tr(P_\mathcal{B}\rho)$, so the expression in the upper bound becomes $\frac{\tr(P_\mathcal{B}\rho)}{\tr(P_A\rho)}$, which is a ratio of probabilities associated to the two subspaces $\mathcal{B}$ and $\mathcal{A}$. More generally, we can further upper bound the bottleneck ratio in a way that depends only on the probabilities of the two subspaces $\mathcal{A}$ and $\mathcal{B}$ evaluated in the steady state $\rho$ using the following

\begin{lemma}\label{lem:bottleneck_bound}
    For any density matrix $\rho$ and projector $P$, the following inequality holds\footnote{The original version of this proof contained an error. We thank Rouchen Ma for pointing it out and for suggesting the elegant alternative proof presented here.}: $||\rho P||_1 \leq \sqrt{\tr(\rho P)}$.
\end{lemma}

\begin{proof}

    We can use H\"older's inequality to write 
    \begin{equation}
        || \rho P ||_1 = || \sqrt{\rho} (\sqrt{\rho} P) ||_1 \leq || \sqrt{\rho} ||_2 \cdot || \sqrt {\rho} P ||_2.
    \end{equation}
    Now, using the definition of the Frobenius norm, we have that $|| \sqrt{\rho} ||_2 = \sqrt{\text{Tr}(\rho)} = 1$ and $|| \sqrt{\rho} P ||_2 = \sqrt{ \text{Tr}(\rho P) }$, completing the proof.
\end{proof}

We can easily generalize from a single time step to any evolution that is given by a product of channels. Let $\{\mathcal{M}_\tau\}_{\tau=1}^t$ be an arbitrary collection of channels and $\sigma$ an arbitrary quantum state. Define $\Delta_t := ||\mathcal{M}_1\ldots\mathcal{M}_t[\sigma] - \sigma||_1$ and $\Delta_{\max} := \max_{1\leq\tau\leq t}(||\mathcal{M}_\tau[\sigma]-\sigma||_1)$. We then have that
\begin{align}
        \Delta_t = ||\mathcal{M}_1\ldots \mathcal{M}_{t-1}(\mathcal{M}_t[\sigma]-\sigma) + (\mathcal{M}_1\ldots \mathcal{M}_{t-1}[\sigma] - \sigma)||_1 \leq \Delta_{\max} + \Delta_{t-1} \leq \ldots \leq t \Delta_{\max},
\end{align}
where we again used that channels are contracting with respect to the trace norm. 

We can also use this to bound the distance from the steady state $\rho$:
\begin{align}
    ||\mathcal{M}_1\ldots\mathcal{M}_t[\sigma] - \rho||_1 = ||(\mathcal{M}_1\ldots\mathcal{M}_t[\sigma] - \sigma) + (\sigma - \rho)||_1 \geq ||\sigma-\rho||_1 - \Delta_t \geq ||\sigma-\rho||_1 - t \Delta_{\max}.
\end{align}
Thus the mixing time, defined as the smallest time needed for an arbitrary initial state to get within trace distance $\varepsilon$ or $\rho$, obeys $t_{\text{mix},\varepsilon} \geq \frac{||\sigma-\rho||_1}{\Delta_{\max}} - 2\varepsilon$. Applying these considerations to Thm.~\ref{thm:quantum_bottleneck_general} we get the following

\begin{corollary}\label{rem:mixing}
    Let $\{\mathcal{M}_\tau\}_{\tau=1}^t$ be a collection of channels sharing the same steady state, $\mathcal{M}_\tau[\rho]=\rho$ and each obeying the condition \autoref{eq:locality_condition} with respect to the same decomposition of the Hilbert space into $\mathcal{A},\mathcal{B}_1,\mathcal{B}_2,\mathcal{C}$. We then have
    \begin{align}
        t_{\text{mix},\varepsilon} \geq \frac{1 - \tr(P_\mathcal{A}\rho) - \varepsilon}{5 \Delta}, \quad \text{where } \, \Delta \equiv \frac{||P_\mathcal{B} \rho||_1}{\operatorname{Tr}(P_\mathcal{A} \rho)}.
    \end{align}
\end{corollary}

\begin{proof}
    We take $\sigma = \rho_\mathcal{A} = P_\mathcal{A} \rho P_\mathcal{A} / \tr(P_\mathcal{A}\rho)$. Applying Thm.~\ref{thm:quantum_bottleneck_general}, we then have that $\Delta_{\max} = \max_\tau||\mathcal{M}_\tau[\sigma]- \sigma||_1 \leq 10\Delta$. What remains is to bound $||\sigma-\rho||_1$ which we do by using~\cite{nielsen2010quantum} $||\sigma-\rho||_1 = 2\max_P\tr(P(\sigma-\rho))$, where the maximization goes over all projectors $P$. This gives
    \begin{equation}
        ||\rho_\mathcal{A} - \rho||_1 = 2 \max_{P} \tr(P(\rho_\mathcal{A}-\rho)) \geq 2 \left(\frac{\tr(P_\mathcal{A}\rho P_\mathcal{A})}{\tr(P_\mathcal{A}\rho)} - \tr(P_\mathcal{A}\rho)\right) = 2(1-\tr(P_\mathcal{A}\rho)),
    \end{equation}
    which completes the proof. 
\end{proof}

\subsection{An alternative bottleneck theorem}

A downside of the bottleneck theorem~\autoref{thm:quantum_bottleneck_general} compared to the usual formulation for Markov chains~\cite{levin2017markov} is the need to split the subspace $\mathcal{B}$ into two parts. Here, we show that it is possible to formulate an alternative version of the theorem that does not rely on this splitting. However, as we will see, the price to pay for this is a somewhat weaker upper bound. 

\begin{theorem}[Quantum Bottleneck Theorem, alternative version]\label{thm:quantum_bottleneck_general_v2}
    Let $\mathcal{M}$ be a quantum channel with steady state $\rho$. Let $\hilbert = \mathcal{A} \oplus \mathcal{B} \oplus \mathcal{C}$ be a decomposition of the Hilbert space into three orthogonal subspaces and let $P_\mathcal{A},P_{\mathcal{B}},P_{\mathcal{C}}$ denote the corresponding projectors. Assume that $\mathcal{M}$ has a Kraus representation such that all Kraus operators $K_i$ satisfy 
    \begin{align}\label{eq:locality_condition_v2}
         P_\mathcal{C} K_i P_\mathcal{A} = P_\mathcal{A} K_i P_\mathcal{C} = 0.
    \end{align}
    Then the projected state, $\rho_A \equiv P_\mathcal{A} \rho P_\mathcal{A} / \tr(P_\mathcal{A}\rho)$ satisfies
    \begin{equation}
        ||\mathcal{M}[\rho_\mathcal{A}] - \rho_\mathcal{A}||_1 \leq c \frac{\sqrt{||P_{\mathcal{B}}\rho||_1}}{\tr(P_\mathcal{A}\rho)},
    \end{equation}
    where $c = 2(\sqrt{3}+1) + 3(\sqrt{5}+1)$ is a constant.
\end{theorem}

\begin{proof}
    The proof again starts by using \autoref{eq:q_bottleneck_inverted} to rewrite the left hand side of the inequality. We then need to bound the following three terms:
    \begin{align}\label{eq:3terms}
        ||\mathcal{M}[P_\mathcal{BC}\rho] - P_\mathcal{BC}\rho||_1, & & ||\mathcal{M}[ \rho P_\mathcal{BC}] -  \rho P_\mathcal{BC}||_1 & & 
        ||\mathcal{M}[P_\mathcal{BC} \rho P_\mathcal{BC}] - P_\mathcal{BC} \rho P_\mathcal{BC}||_1.
    \end{align}

Let us focus on the first of these, the other two can be bounded in a similar way. Multiplying by $1\!\!1 = P_\mathcal{AB} + P_\mathcal{C}$ from the left and using~\autoref{eq:locality_condition_v2} we find that
\begin{equation}\label{eq:sum_bound}
    ||\mathcal{M}[P_\mathcal{BC}\rho] - P_\mathcal{BC}\rho||_1 = ||P_\mathcal{AB} \mathcal{M}[P_\mathcal{BC}\rho] - P_\mathcal{B}\rho||_1 \leq ||P_{\mathcal AB}\mathcal{M}[P_\mathcal{BC}\rho]||_1 + ||P_\mathcal{B}\rho||_1,
\end{equation}
where in the last step we again made use of the fact that multiplying by a projector does not increase the $1$-norm. 

To show the desired statement, it remains to bound $||P_\mathcal{AB}\mathcal{M}[P_\mathcal{BC}\rho]||_1$ in terms of $||P_{\mathcal B}\rho||_1$. We do so as follows. On the one hand, note that
\begin{equation}\label{eq:bound1}
    ||\mathcal{M}[P_\mathcal{BC}\rho]||_1 \leq ||P_\mathcal{BC}\rho||_1 \leq ||P_\mathcal{B}\rho||_1 + ||P_\mathcal{C}\rho||_1.    
\end{equation}
On the other hand, we will need the following statement, which we prove below in Lemma~\autoref{lem:splitting}:
\begin{equation}\label{eq:bound2}
    ||\mathcal{M}[P_\mathcal{BC}\rho]||_1 \geq \sqrt{||P_\mathcal{AB}\mathcal{M}[P_\mathcal{BC}\rho]||_1^2 + ||P_\mathcal{C}\mathcal{M}[P_\mathcal{BC}\rho]||_1^2} = \sqrt{||P_\mathcal{AB}\mathcal{M}[P_\mathcal{BC}\rho]||_1^2 + ||P_\mathcal{C}\rho||_1^2}.
\end{equation}
Combining~\autoref{eq:bound1} and~\autoref{eq:bound2} we end up with the following inequality:
\begin{equation}
    ||P_\mathcal{AB}\mathcal{M}[P_\mathcal{BC}\rho]||_1^2 \leq ||P_\mathcal{B}\rho||_1 (2||P_\mathcal{C}\rho||_1 + ||P_\mathcal{B}\rho||_1) \leq ||P_\mathcal{B}\rho||_1 (2 + ||P_\mathcal{B}\rho||_1).
\end{equation}
We thus conclude that the RHS of~\autoref{eq:sum_bound} is upper bounded by 
\begin{equation}
    ||P_{\mathcal AB}\mathcal{M}[P_\mathcal{BC}\rho]||_1 + ||P_\mathcal{B}\rho||_1 \leq  ||P_\mathcal{B}\rho||_1 + \sqrt{||P_\mathcal{B}\rho||_1 (2 + ||P_\mathcal{B}\rho||_1)} \leq (\sqrt{3}+1) \sqrt{||P_\mathcal{B}\rho||_1}
\end{equation}
where in the last step we used that $||P_\mathcal{B}\rho||_1 \in [0,1]$.

The same upper bound applies to the second term in~\autoref{eq:3terms}. The bound of the third term follows from a similar argument as follows. First, we insert identities to obtain
\begin{align}\label{eq:lastbound}
    &||\mathcal{M}[P_\mathcal{BC} \rho P_\mathcal{BC}] - P_\mathcal{BC} \rho P_\mathcal{BC}||_1 = \nonumber \\ &= ||P_\mathcal{AB}\mathcal{M}[P_\mathcal{BC} \rho P_\mathcal{BC}]P_\mathcal{C} + P_\mathcal{C}\mathcal{M}[P_\mathcal{BC} \rho P_\mathcal{BC}]P_\mathcal{AB} + P_\mathcal{AB}\mathcal{M}[P_\mathcal{BC} \rho P_\mathcal{BC}]P_\mathcal{AB} - P_\mathcal{B} \rho P_\mathcal{C} - P_\mathcal{C} \rho P_\mathcal{B} - P_\mathcal{B} \rho P_\mathcal{B}||_1 \nonumber \\
    &\leq ||P_\mathcal{AB}\mathcal{M}[P_\mathcal{BC} \rho P_\mathcal{BC}]P_\mathcal{C}||_1 + ||P_\mathcal{C}\mathcal{M}[P_\mathcal{BC} \rho P_\mathcal{BC}]P_\mathcal{AB}||_1 + ||P_\mathcal{AB}\mathcal{M}[P_\mathcal{BC} \rho P_\mathcal{BC}]P_\mathcal{AB}||_1 + 3 ||P_\mathcal{B}\rho||_1.
\end{align}
Let $x,y,z$ denote the first three terms on the RHS of~\autoref{eq:lastbound}. We need to bound $x+y+z$. By the properties of $p$-norms, this is upper bounded by $\sqrt{3(x^2+y^2+z^2)}$. We can further upper bound this by again applying Lemma~\autoref{lem:splitting} which gives
    \begin{equation}
        x^2 + y^2 + z^2 + ||P_\mathcal{C} \rho P_\mathcal{C}||_1^2 \leq ||\mathcal{M}[P_\mathcal{BC}\rho P_\mathcal{BC}]||_1^2 \leq ||P_\mathcal{BC}\rho P_\mathcal{BC}||_1^2 \leq \left(||P_\mathcal{C} \rho P_\mathcal{C}||_1 + 3 ||P_\mathcal{B} \rho ||_1\right)^2,
    \end{equation}
    which implies that $x^2 + y^2 + z^2 \leq 6 ||P_\mathcal{B}\rho||_1 + 9 ||P_\mathcal{B}\rho||_1^2$. Finally, plugging this back into~\autoref{eq:lastbound} we get that
    \begin{equation}
        ||\mathcal{M}[P_\mathcal{BC} \rho P_\mathcal{BC}] - P_\mathcal{BC} \rho P_\mathcal{BC}||_1 \leq 3 (\sqrt{||P_\mathcal{B} \rho ||_1 (2 + 3 ||P_\mathcal{B} \rho ||_1)} + ||P_\mathcal{B} \rho ||_1) \leq 3 (\sqrt{5}+1) \sqrt{||P_\mathcal{B} \rho ||_1}. 
    \end{equation}
    Combining all the terms we end up with the desired statement.

\end{proof}

\begin{lemma}\label{lem:splitting}
    Let $P_1 + P_2 = 1\!\!1$ be two complementary orthogonal projectors. Then for any matrix $A$,
    \begin{equation}
        ||A||_1^2 \geq ||P_1 A||_1^2 + ||P_2 A||_1^2.
    \end{equation}
\end{lemma}

\begin{proof}
    We begin by rewriting the right hand side using a pair of polar decompositions, $P_i A = X_i U_i$ where the $U_i$ are unitary and $X_i = \sqrt{P_i A A^\dagger P_i}$ so that $\text{tr}(X_i) = ||P_iA||_1$. We also define positive coefficients $t_i = \frac{||P_i A||_1}{\sqrt{||P_1 A||_1^2 + ||P_2 A||_1^2}}$ and the matrix $B = t_1 P_1 U_1 + t_2P_2U_2$. We then find that
    \begin{equation}
        \text{tr}\left(AB^\dagger\right) = t_1\text{tr}(P_1AU_1^\dagger) +  t_2\text{tr}(P_2AU_2^\dagger) = \frac{||P_1 A||_1^2 + ||P_2 A||_1^2}{\sqrt{||P_1 A||_1^2 + ||P_2 A||_1^2}} = \sqrt{||P_1 A||_1^2 + ||P_2 A||_1^2}.
    \end{equation}
    Now, by the duality property of the trace norm, we have that 
    \begin{equation}
        ||P_1 A||_1^2 + ||P_2 A||_1^2 = |\text{tr}(AB^\dagger)|^2 \leq ||A||_1^2 \cdot ||B||^2.
    \end{equation}

    Finally, we bound $||B||$ using
    \begin{equation}
        ||B|| = \sup_{||v||=1} (||Bv||) = \sup_{||v||=1} (\sqrt{||P_1Bv||^2 + ||P_2Bv||^2}) \leq \sqrt{||P_1B||^2 + ||P_2B||^2}.
    \end{equation}
    In the present case, plugging in the definition of $B$, this yields
    \begin{equation}
        ||B||^2 \leq t_1^2 ||P_1U_1||^2 + t_2^2 ||P_2U_2||^2 \leq t_1^2 + t_2^2 =1,
    \end{equation}
    where we used the submultiplicativity of the operator norm, $||P_i U_i|| \leq ||P_i|| \cdot ||U_i|| = 1$. This completes the proof.
\end{proof}

\begin{remark}
    Note that the weakening of the bound---from $||P_\mathcal{B} \rho ||_1$ to $\sqrt{||P_\mathcal{B} \rho ||_1}$ in the numerator of the bottleneck ratio in Thm.~\autoref{thm:quantum_bottleneck_general_v2}---originated from the bound in Lemma~\autoref{lem:splitting} which relates the squares of the norms of parts of a matrix to that of the whole. If we repeated the proof for a classical Markov chain, we would instead need to bound $||MP_{AB}\ket{\pi}||_1$ which is \emph{equal} to $||P_{AB}MP_{AB}\ket{\pi}||_1 + ||P_C MP_{AB}\ket{\pi}||_1$, implying $||P_{AB}MP_{AB}\ket{\pi}||_1 \leq ||P_B \ket{\pi}||_1$. As a result, we end up with the \emph{same} bottleneck ratio as in the classical theorem~\autoref{thm:classical_bottleneck_general} but without the need to split the region $B$ into two parts.
\end{remark}

\subsection{Local channels}

We now specify to the case when $\mathcal{H}$ is the Hilbert space of $n$ qubits and consider channels that obey certain (mild) locality conditions. To define these, consider the set of all Pauli operators (products of Pauli matrices acting on different qubits), acting on at most $r$ qubits, defined as 
\begin{equation}
    \mathcal{P}_r = \{S\mid S \in \mathcal{P},\, |\text{supp}(S)| \leq r\} \subseteq \mathcal{P},
\end{equation} 
with $\text{supp}(S)$ being the support of $S$, i.e.\ the set of qubits it acts on non-trivially. Using this, we can now define neighborhoods and boundaries in Hilbert space as follows.

\begin{definition}
    The $r$-\textbf{Hilbert space neighborhood} of the subspace $\mathcal{V}$ is defined as 
\begin{equation}
    \mathcal{B}_r(\mathcal{V}) \equiv \operatorname{span}\{S\ket{\psi}\mid \ket{\psi} \in \mathcal{V}, S \in \mathcal{P}_r\}.
\end{equation}
Note that $\mathcal{B}_r(\mathcal{V})$ is itself a subspace containing $\mathcal{V}$, $\mathcal{V} \subseteq \mathcal{B}_r(\mathcal{V})$, by construction. The $r$-\textbf{boundary} of $\mathcal{V}$ is defined as the orthogonal complement of $\mathcal{V}$ within $\mathcal{B}_r(\mathcal{V})$. In terms of the corresponding projectors: $P_{\partial_r \mathcal{V}} = P_{\mathcal{B}_r(\mathcal{V})} (1\!\!1 - P_{\mathcal{V}})$.
\end{definition}

\begin{remark}
    Since $\mathcal{V}$ is contained within $\mathcal{B}_r(\mathcal{V})$, their projectors commute and $P_{\mathcal{B}_r(\mathcal{V})} = P_{\mathcal{V}} + P_{\partial_r \mathcal{V}}$.
\end{remark}

    The subspace $\mathcal{B}_r(\mathcal{V})$ can be thought of as the states ``within distance $r$'' of $\mathcal{V}$. This interpretation is facilitated by the following:
\begin{proposition}
    For any subspace $\mathcal{V}$ and $r,s \geq 0$:
    \begin{equation}
        \mathcal{B}_r(\mathcal{B}_{s}(\mathcal{V})) = \mathcal{B}_{r+s}(\mathcal{V}).
    \end{equation}
\end{proposition}

    \begin{proof}
     First, we will show that $\mathcal{B}_r(\mathcal{B}_{s}(\mathcal{V})) \subseteq \mathcal{B}_{r+s}(\mathcal{V})$. A generic element of $\mathcal{B}_s(\mathcal{V})$ takes the form $\ket{\psi'} = \sum_{S' \in \mathcal{P}_s} c'_{S'} S' \ket{\psi_{S'}}$ where $\ket{\psi_{S'}} \in \mathcal{V}$. We then have that a generic element of $\mathcal{B}_r(\mathcal{B}_{s}(\mathcal{V}))$ can be written as
    \begin{equation}
        \ket{\psi''} = \sum_{S''\in\mathcal{P}_r} c''_{S''} S'' \ket{\psi'_{S''}} = \sum_{S''\in \mathcal{P}_r} \sum_{S' \in \mathcal{P}_s} c''_{S''} c'_{S',S''} S'' S' \ket{\psi_{S',S''}}.
    \end{equation}
    Here $\ket{\psi_{S',S''}} \in \mathcal{V}$ and since $S'' S' \in \mathcal{P}_{r+s}$, we have that $\ket{\psi''} \in \mathcal{B}_{r+s}$.

    To see containment in the other direction, $\mathcal{B}_{r+s}(\mathcal{V}) \subseteq \mathcal{B}_r(\mathcal{B}_{s}(\mathcal{V}))$, write a generic element of $\mathcal{B}_{r+s}(\mathcal{V})$ as $\ket{\psi} = \sum_{S \in \mathcal{P}_{r+s}} c_S \ket{\psi_S}$. If $S\in\mathcal{P}_{r+s}$, then it can be split (in a non-unique way) as $S = S'' S'$, where $S' \in \mathcal{P}_s$ and $S'' \in \mathcal{P}_r$. Choose such a splitting for every $S$; then $S' \ket{\psi_{S=S'S''}} \in \mathcal{B}_s(\mathcal{V})$ and consequently $S\ket{\psi_S} = S'' S'\ket{\psi_S} \in \mathcal{B}_r(\mathcal{B}_{s}(\mathcal{V}))$. 
\end{proof}

\begin{definition}
    Let $\mathcal{M}$ be a quantum channel. We say that $\mathcal{M}$ is $r$-\textbf{local} if it has a Kraus representation, $\mathcal{M}[\rho] = \sum_i K_i \rho K_i^\dagger$, such that all the Kraus operators $K_i$ act on at most $r$ qubits. 
\end{definition}

\begin{remark}
     If $K$ acts on at most $r$ qubits, it can be written as $K = \sum_{S \in \mathcal{P}_r} a_{S}S$ for some set of coefficients $a_{S}$.
\end{remark}

\begin{theorem}[Quantum Bottleneck Theorem, local channels]
\label{thm:bottleneck}
    Let $\mathcal{M}$ be an $r$-local quantum channel with steady state $\rho$. Let $\mathcal{V}$ be a subspace and define the quantum bottleneck ratio as
    \begin{equation}
        \frac{||P_{\partial_{2r}\mathcal{V}}\rho||_1}{\tr(P_{\mathcal{V}}\rho)} = \frac{||\rho P_{\partial_{2r}\mathcal{V}}||_1}{\tr(P_{\mathcal{V}}\rho)} = \Delta.
    \end{equation}
    Define the projected state $\rho_{\mathcal{V}} = P_{\mathcal{V}} \rho P_{\mathcal{V}} / \tr(P_{\mathcal{V}}\rho)$. Then $\rho_{\mathcal{V}}$ is an approximate steady state of $\mathcal{M}$ in the following sense:
    \begin{equation}
        ||\mathcal{M}[\rho_{\mathcal{V}}] - \rho_{\mathcal{V}}||_1 \leq 10 \Delta.
    \end{equation}
\end{theorem}

\begin{proof}
        The claim follows directly from Thm.~\ref{thm:quantum_bottleneck_general} with an appropriate identification of subspaces. In particular, we split the $2r$-neighborhood of $\mathcal{V}$ region into two halves by using $\mathcal{B}_{2r}(\mathcal{V}) = \mathcal{B}_r(\mathcal{B}_r(\mathcal{V}))$ and denote by $P_\mathcal{A}$, $P_{\mathcal{A}\mathcal{B}_1}$, $P_{\mathcal{A}\mathcal{B}_1\mathcal{B}_2} = P_{\mathcal{A}\mathcal{B}}$ the projectors onto $\mathcal{V}$, $\mathcal{B}_r(\mathcal{V})$ and $\mathcal{B}_{2r}(\mathcal{V})$ respectively. This induces a decomposition of the Hilbert space into four orthogonal subspaces $1\!\!1 = P_\mathcal{A} + P_{\mathcal{B}_1} + P_{\mathcal{B}_2} + P_\mathcal{C}$, where $P_{\mathcal{B}_1} = P_{\mathcal{A}\mathcal{B}_1}(1\!\!1-P_\mathcal{A})$, $P_{\mathcal{B}_2} = P_{\mathcal{A}\mathcal{B}_1\mathcal{B}_2}(1\!\!1-P_{\mathcal{A}\mathcal{B}_1})$ and $P_{\mathcal{C}} = 1\!\!1-P_{\mathcal{A}\mathcal{B}_1\mathcal{B}_2}$. 

        Consider now the Kraus representation of $\mathcal{M}$ with Kraus operators $i$ that are assumed to be $r$-local operators and let $\ket{\psi}\in\mathcal{V}$. Then we have that $K_i\ket{\psi} = \sum_{S\in\mathcal{P}_r} a_{S}^{(i)}S\ket{\psi} \in \mathcal{B}_r(\mathcal{V})$. This implies $P_{\mathcal{B}_2\mathcal{C}}K_i P_\mathcal{A} = 0$. A similar argument also yields $P_{\mathcal{A}\mathcal{B}_1} K_i P_\mathcal{C} = 0$. Thus $\mathcal{M}$ fulfills the conditions of Thm.~\ref{thm:quantum_bottleneck_general} and the claim follows. 
\end{proof}

We can generalize the theorem to a more general set of channels, which fulfill a less stringent locality condition.

\begin{definition}
    We say that the channel $\mathcal{M}$ is $f$-\textbf{quasi-local}, for a decaying function $\lim_{r\to\infty}f(r)=0$, if $\forall r \geq 1$ there exists an $r$-local channel $\mathcal{M}^{(r)}$ such that $||\mathcal{M}-\mathcal{M}^{(r)}||_\diamond \leq f(r)$. 
\end{definition}

For these, we can prove a different bottleneck theorem:

\begin{theorem}[Quantum Bottleneck Theorem, quasi-local channels]
    Let $\mathcal{M}$ be an $f$-quasi-local channel with steady state $\rho$. Let $\mathcal{V}$ be a subspace and $\rho_{\mathcal{V}} = P_{\mathcal{V}}\rho P_{\mathcal{V}} / \tr(P_{\mathcal{V}}\rho)$ be the steady state projected onto this subspace. Define 
    \begin{equation}
    \Delta(s) \equiv \frac{||P_{\partial_{2s}\mathcal{V}}\rho||_1}{\tr(P_{\mathcal{V}}\rho)}
    \end{equation}
    We then have 
    \begin{equation}
        ||\mathcal{M}[\rho_{\mathcal{V}}] - \rho_{\mathcal{V}}||_1 \leq \inf_{1\leq s \leq n/2}\left(10\Delta(s) + f(s)\right).
    \end{equation}    
\end{theorem}

\begin{proof}
    Fix $s \geq 1$ and write $\mathcal{M} = \mathcal{M}^{(s)} + \delta \mathcal{M}^{(s)}$, where $\mathcal{M}^{(s)}$ is the closest $s$-local approximation of $\mathcal{M}$. We then have
    \begin{equation}
        ||\mathcal{M} [\rho_\mathcal{V}] - \rho_\mathcal{V}||_1 = ||\mathcal{M}^{(s)}[\rho_\mathcal{V}] - \rho_\mathcal{V} + \delta \mathcal{M}^{(s)}[\rho_\mathcal{V}]||_1 \leq ||\mathcal{M}^{(s)}[\rho_\mathcal{V}] - \rho_\mathcal{V}||_1 + ||\delta M^{(s)}[\rho_\mathcal{V}]||_1 \leq 10\Delta(s) +f(s),
    \end{equation}
    where we used \cref{thm:quantum_bottleneck_general} to bound the first term and we used that $||\mathcal{M}[\sigma]||_1 \leq ||\mathcal{M}||_\diamond ||\sigma||_1$ for any super-operator $\mathcal{M}$. Since the upper bound applies for any choice of $s$, the statement follows.
\end{proof}

\section{Stability of commuting projector Hamiltonians with extensive energy barriers}

We consider Hamiltonians of the form $H = H_0 + V$ on $n$ qubits. The unperturbed Hamiltonian, $H_0 = \sum_k \Pi_k$, is a sum of mutually commuting projectors $\Pi_k$ and we assume both that each projector acts on finitely many qubits and each qubit appears in finitely many terms of $H_0$. $V$ is a perturbation that is a sum of local terms and whose operator norm satisfies $||V|| \leq g n$ for some constant perturbation strength $g$. 

We will assume that $H_0$ exhibits linear energy barriers around various local energy minima. In particular, we consider the setup of Thm.~\ref{thm:bottleneck} above. Let $\mathcal{V}$ be a subspace of $\hilbert$ and $\partial_{r(n)}\mathcal{V}$ its boundary for some function $r(n)$ diverging super-logarithmically in $n$; in the following we simplify notation and refer to this boundary subspace simply as $\partial\mathcal{V}$. Let $E_{0,\min}(\mathcal{V}) \equiv \min_{\ket{\psi}\in\mathcal{V}}(\langle\psi|H_0|\psi\rangle)$ denote the minimal energy in $\mathcal{V}$. We then say that $H_0$ has an extensive energy barrier surrounding $\mathcal{V}$ if 
\begin{equation}\label{eq:def_lin_barrier}
    E_{0,\min}(\partial\mathcal{V}) - E_{0,\min}(\mathcal{V}) \geq \kappa n    
\end{equation}
for some constant $\kappa > 0$. In the following, we will write $E_{0,\min}(\mathcal{V}) = \epsilon n$ where $\epsilon$ defines an energy density associated to the subspace $\mathcal{V}$ (which might be zero). Our goal is to argue that whenever \autoref{eq:def_lin_barrier} holds, and $\epsilon$ is sufficiently small, the bottleneck ratio $\Delta$ is exponentially small in $n$ for sufficiently large $\beta$ and small perturbation strength $g$. 

Before discussing the proof, let us review some examples that satisfy this property, originating from classical and quantum error correcting codes. In these cases, a property of codes called \emph{expansion} guarantees the existence of extensive energy barriers for various choices of subspaces $\mathcal{V}$, including ones at finite energy density, with $\epsilon > 0$. Indeed, for certain codes, the number of distinct subspaces having such barriers is exponentially large in $n$ and increases with $\epsilon$; see Refs. \cite{placke2025expansion,placke2024tqsg}.

\begin{example}[Classical expander codes]
    We can construct Hamiltonians out of classical error correcting codes by taking $\Pi_k = (1\!\!1-C_k)/2$ where $C_k = \prod_{i \in \delta(k)} Z_i$ correspond to the parity check of a classical code (with $Z_i$ the diagonal Pauli matrix acting on qubit $i$), defined by its support $\delta(k)$ (see e.g.~\cite{rakovszky2023physics} for more details). Eigenstates of $H_0$ are diagonal in the computational basis and their energy is given by the number of violated checks. We can represent each configuration by a binary vector $\vec x$ and denote its energy by $E_0(\vec x)$. 
    
    There exist examples, including various constructions of good classical codes~\cite{sipser_spielman1996} that satisfy $(\delta,\gamma)$-\textbf{expansion}
    \begin{equation}
        |\vec x| \leq \delta n \Rightarrow E_0(\vec x) \geq \gamma |\vec x|,
    \end{equation}
    meaning that the energy grows linearly with the Hamming weight $|\vec x|$. This already implies an extensive energy barrier around any of the ground states of $H_0$ (the code's \emph{codewords}). We can take $\vec x_0$ to be such a codeword, satisfying $E_0(\vec x_0) = 0$ and choose $\mathcal{V}$ to be the subspace spanned by configurations that are within Hamming distance $\leq \delta n /2$ of $\vec x_0$. If we then take $r(n) \leq \delta n /2$ in the definition of $\partial\mathcal{V}$, it will satisfy Eq.~\eqref{eq:def_lin_barrier} with  $\kappa = \gamma\delta /2$.

    In fact, expansion also implies the existence of extensive barriers surrounding other low-energy configurations. It follows from the triangle inequality for the Hamming weight that~\cite{placke2024tqsg} 
    \begin{equation}
        |\vec x + \vec x_0| \leq \delta n \Rightarrow E_0(\vec x) - E_0(\vec x_0) \geq \gamma |\vec x + \vec x_0| - 2E_0(\vec x_0).
    \end{equation}
    Therefore, when $E_0(\vec x_0)$ is sufficiently small, $\vec x_0$ is still surrounded by an extensive energy barrier. For example, if $E_0(\vec x_0) \leq \gamma\delta n /4$, we can take $\mathcal{V}$ to be a subspace defined by configurations within Hamming distance $\frac{3}{4}\delta n$ and $r(n) \leq \delta n /4$ to get an extensive barrier with $\kappa = \delta\gamma/4$.
\end{example}

\begin{example}[Good quantum CSS codes]
    Similar expansion properties are known to hold for certain \emph{quantum} error correcting codes. In these cases, there are two sets of parity checks, $B_k = \prod_{i \in \delta_Z(k)}Z_i$ and $A_{k'} = \prod_{i \in \delta_X(k')}X_i$, all mutually commuting, and we can combine them into a stabilizer Hamiltonian $H_0 = \sum_k \frac{1\!\!1 + B_k}{2} + \sum_{k'} \frac{1\!\!1 + A_{k'}}{2}$ whose ground states comprise the codestates of the corresponding quantum code. Other eigenstates can be written as 
    \begin{equation}\label{eq:code_eigstates}
    \ket{\vec x, \vec z} \equiv \left(\prod_{i:x_i=1} X_i\right) \left(\prod_{i':z_{i'}=1} Z_{i'}\right) \ket{\psi_0},   
    \end{equation}
    where $\ket{\psi_0}$ is a groundstate of $H_0$ and $\vec x, \vec z$ are a pair of binary vectors. Energy with respect to $H_0$ is given by the number of checks of each type that is violated in such a state and can be written as $E_0(\vec x, \vec z) = E_0^Z(\vec x) + E_0^X(\vec z)$. 

    Importantly, the labeling of eigenstates by $\vec x, \vec z$ is not unique: the operator $\prod_{i:x_i=1} X_i$ can be deformed by multiplying it with the checks $A_i$ without changing its effect on $\ket{\psi_0}$ and similarly for the $Z$-type operators. One thus defines a ``reduced weight'' $|\vec x|_{\rm red}$ as the minimum of the Hamming weight $|\vec x|$ over all equivalent choices of $\vec x$. Code expansion is then defined in terms of this reduced weight as~\cite{leverrier2015quantum_expander,dinur2023qldpc}
    \begin{align}
        |\vec x|_{\rm red} \leq \delta n &\Rightarrow E^Z_0(\vec x) \geq \gamma |\vec x|_{\rm red}, \nonumber \\
        |\vec z|_{\rm red} \leq \delta n &\Rightarrow E^X_0(\vec z) \geq \gamma |\vec z|_{\rm red}.     
    \end{align}
    The discussion from classical codes then carries over using the modified definition of distance. I.e., we can fix a ground state $\ket{\psi_0}$ and define the subspace $\mathcal{V}$ around it to be spanned by all eigenstates $\ket{\vec x,\vec z}$ of the form Eq.~\eqref{eq:code_eigstates} with $|\vec x|_{\rm red},|\vec z|_{\rm red} \leq \delta n /2$. Since any operator can be written as a linear combination of products of Pauli $X$ and $Z$ operators, multiplying a vector in $\mathcal{V}$ with an $r$-local operator results in a state that is a superposition of eigenstates with $|\vec x|_{\rm red},|\vec z|_{\rm red} \leq \delta n /2 + r$. Therefore if we take $r(n) \leq \delta n /2$ in the definition of $\partial\mathcal{V}$, we will satisfy Eq.~\eqref{eq:def_lin_barrier} with $\kappa = \gamma\delta$. 
    
    Finally, using the triangle inequality, now for $\abs{\bullet}_{\rm red}$, we get~\cite{placke2024tqsg} 
    \begin{align}
        |\vec x + \vec x_0|_{\rm red} \leq \delta n &\Rightarrow E^Z_0(\vec x) - E^Z_0(\vec x_0) \geq \gamma |\vec x + \vec x_0|_{\rm red} - 2E^Z_0(\vec x_0), \nonumber \\
        |\vec z + \vec z_0|_{\rm red} \leq \delta n &\Rightarrow E^X_0(\vec z) - E^X_0(\vec z_0) \geq \gamma |\vec z + \vec z_0|_{\rm red} - 2E^X_0(\vec x_0).
    \end{align}
    We can therefore take some reference state $\ket{\vec x_0,\vec z_0} $ with energy $E_0(\vec x_0,\vec z_0) \leq \gamma\delta n/4$ and take $\mathcal{V}$ to be defined by eigenstates of the form $\ket{\vec x_0 +\vec x,\vec z_0 +\vec z}$ with $|\vec x|_{\rm red},|\vec z|_{\rm red} \leq 3\delta n /4$ and choose $r(n) \leq \delta n / 4$ to achieve Eq.~\eqref{eq:def_lin_barrier} with $\kappa = \delta\gamma/2$.
\end{example}

We now want to prove that extensive barriers with respect to $H_0$ imply an exponentially small bottleneck ratio even in the perturbed model $H_0 + V$, assuming the perturbation is weak and the temperature is low. 

To bound the bottleneck ratio, we will make use of Lemma~\ref{lem:bottleneck_bound} to write 
\begin{equation}\label{eq:bottleneck_expander}
    \Delta \leq \frac{\sqrt{\tr(\rho_G P_{\partial\mathcal{V}})}}{\tr(\rho_G P_\mathcal{V})} = Z^{1/2} \frac{\sqrt{\tr(e^{-\beta H} P_{\partial\mathcal{V}})}}{\tr(e^{-\beta H} P_\mathcal{V})},
\end{equation}
where $Z = \tr(e^{-\beta H})$ is the partition function. The denominator is easy to bound:
\begin{equation}\label{eq:weight_in_V}
    \tr(e^{-\beta H} P_\mathcal{V}) \geq \max_{\ket{\psi}\in\mathcal{V}} \langle\psi|e^{-\beta H}|\psi\rangle \geq e^{-\beta \min_{\ket{\psi}\in\mathcal{V}}(\langle \psi|H|\psi\rangle)} \geq e^{-\beta (\epsilon + g) n},
\end{equation}
where we used the convexity of the exponential function and the fact that the perturbation $V$ changes the energy of any state by at most $gn$. 

The more difficult part is upper bounding $\tr(e^{-\beta H} P_{\partial\mathcal{V}})$. The problem here is that, even though we know that the \emph{average} energy of any state in $\partial \mathcal{V}$ is large (i.e., larger than $(\epsilon + \kappa - g)n$, thanks to \autoref{eq:def_lin_barrier}), this in itself does not give any useful bound on the expectation value of $e^{-\beta H}$ in this state, which could be dominated by contributions from the small components of the wavefunction on eigenstates with energies much smaller than the average. Luckily, we can adapt a proof from Ref. \cite{yin2024eigenstate} to bound these tails.

\begin{lemma}\label{lem:tail_bound}
    Let $\ket{\psi}$ be an eigenstate of $H = H_0 + V$ with energy $E < \epsilon_1 n$ and let $Q_>$ be a projector onto the subspace $\hilbert_>$, spanned by eigenstates of $H_0$ with energies $\geq \epsilon_2 n$ with $\epsilon_2 > \epsilon_1 + 4g$. Let each qubit be part of at most $w_0$ terms in $H_0$ and each term in $V$ act on at most $w_1$ qubits. Then 
    \begin{equation}
        ||Q_> \ket{\psi}|| \leq e^{-\lambda(g) n}, \quad \text{ with } \lambda(g) \geq \frac{\epsilon_2 -\epsilon_1}{2w_0w_1} \ln\left(\frac{\epsilon_2-\epsilon_1}{2g}\right).
    \end{equation}
\end{lemma}

\begin{proof}
    Let $\ket{\psi_0}$ be an eigenstate of $H_0$ with energy $E_0$. Such an eigenstate is characterized by some pattern of excitations, i.e. the set of projectors $\Pi_k$ that have eigenvalue $1$ in $\ket{\psi_0}$. Write $V = \sum_l V_l$ where the terms $V_l$ act on at most $w_1$ qubits. Since each qubit is part of at most $w_0$ terms in $H_0$, therefore the state $V\ket{\psi_0}$ is contained within the subspace spanned by the eigenstates of $H_0$ with energies in the interval $[E_0-w_0w_1,E_0+w_0w_1]$. The perturbation $V$ thus only connects eigenstates of $H_0$ with nearby energy.
        
 We now divide the Hilbert space into a sum of orthogonal subspaces, 
    \begin{equation}\label{eq:split_hilbert}
    \hilbert = \hilbert_> \oplus \bigoplus_{q=1}^{q_*}\hilbert_q \oplus \hilbert_<,
    \end{equation}
     where the different subspaces are defined as follows %
    \begin{itemize}
    	\item $\hilbert_>$ is spanned by eigenstates of $H_0$ with eigenvalues $E_0 \geq \epsilon_2 n$
	\item $\hilbert_<$ is spanned by eigenstates of $H_0$ with eigenvalues $ E_0 < \left(\frac{\epsilon_2 - \epsilon_1}{2} +2g\right) n$.
	\item  $\hilbert_q$ is spanned by eigenstates of $H_0$ with eigenvalues $E_0 \in [E(q),E(q)+\Delta E)$, where $\Delta E > w_0w_1$ is an $O(1)$ constant and $E(q) = E(1) + (q-1) \Delta E$, with $E(1) = \frac{\epsilon_2 + \epsilon_1}{2} n + 2gn$ and $E(q_*) + \Delta E = \epsilon_2 n$. Note that $q_* \Delta E = \left(\frac{\epsilon_2-\epsilon_1}{2} - 2g\right)n$.
    \end{itemize}
     In this decomposition, $H_0$ is block-diagonal and $V$ is block-tri-diagonal, only connecting neighboring subspaces, so we can write
     \begin{align}\label{eq:tridiag}
     	H_0 =
     	\begin{pmatrix}
	 	h_> & 0 & 0 & 0 & \ldots \\
		0 & h_{q_*} & 0 & 0 & \ldots  \\
            \vdots & \vdots & \ddots & \vdots & \vdots \\
            \ldots & 0 & 0 & h_1 & 0 \\
            \ldots & 0 & 0 & 0 & h_< 
	\end{pmatrix}; & & 
	V = 
     	\begin{pmatrix}
	 	v_> & v_{>,q_*} & 0 & 0 & \ldots \\
		v_{>,q_*}^\dagger & v_{q_*} & v_{q_*,q_*-1} & 0 &  \ldots \\
            \vdots & \vdots &  \ddots & \vdots & \vdots  \\
            \ldots & 0 & v_{2,1}^\dagger & v_1 & v_{1,<} \\
            \ldots & 0 & 0 & v_{1,<}^\dagger & v_< 
	\end{pmatrix}.
     \end{align}  

    Now consider an eigenstate $\ket{\psi}$ of $H = H_0 +V$ with energy $E$, which we decompose according to Eq.~\eqref{eq:split_hilbert} as 
    \begin{equation}
    \ket{\psi} = c_> \ket{\psi_>} + \sum_q c_q \ket{\psi_q}  + c_< \ket{\psi_<},
    \end{equation}
    where the states $\ket{\psi_q}$, $\ket{\psi_>}$ and $\ket{\psi_<}$ are normalized and the coefficients $c_q$, $c_>$ and $c_<$ can be chosen to be real and non-negative. Using Eq.~\eqref{eq:tridiag}, the eigenvalue equation implies 
    \begin{equation}
    	E c_> \ket{\psi_>} = c_> (h_> + v_>) \ket{\psi_>} + c_{q_*} v_{>,q_*} \ket{\psi_{q*}}.
    \end{equation}
    Multiplying by $\bra{\psi_{>}}$ and using $||V|| \leq g n$, along with the definition of $\hilbert_>$, we find
    \begin{equation}
    c_{q_*} gn \geq-c_{q_*} \langle \psi_> | v_{>,q_*} | \psi_> \rangle = c_> \left(  \langle \psi_> | h_> | \psi_> \rangle + \langle \psi_> | v_> | \psi_> \rangle - E \right) \geq c_> (\epsilon_2 n - gn - E).
    \end{equation}
    Rearranging this we get that $c_>$ is upper bounded by $c_{q_*}$ as
    \begin{equation}
    c_> \leq \frac{gn}{\epsilon_2 n - gn - E} c_{q_*} \leq c_{q_*}.
    \end{equation}
    
    We can now iterate this bound to get an increasingly strong upper bound on $c_>$. In particular, for any $1 \leq q \leq q_*$ the eigenvalue equation yields (identifying $\hilbert_>$ with $\hilbert_{q_*+1}$ and $\hilbert_<$ with $\hilbert_0$):
    \begin{equation}\label{eq:eigval_eq}
        E c_q \ket{\psi_q} = c_q (h_q + v_q) \ket{\psi_q} + c_{q+1} v_{q+1,q}^\dagger \ket{\psi_{q+1}} + c_{q-1} v_{q,q-1}^\dagger \ket{\psi_{q-1}}
    \end{equation}
    Let us now assume that $c_{q+1} \leq c_q$. Then multiplying Eq.~\eqref{eq:eigval_eq} by $\bra{\psi_q}$ and rearranging we get the following inequality:
    \begin{equation}\label{eq:coeff_ineq}
        c_q \leq \frac{g n} {(E(q)- gn - E) - c_{q+1} g n} c_{q-1} \leq \frac{g n} {E(q)- 2gn - E} c_{q-1} \leq c_{q-1}.
    \end{equation}
    We can thus iterate this bound all the way to $q=1$ to get:
    \begin{equation}
        c_> \leq \prod_{q=1}^{q_*} \frac{gn}{E(q) - 2gn - E} c_< \leq  \left(  \frac{2gn}{E(1) - 2gn - E} \right)^{q_*} = \left(  \frac{2g}{\epsilon_2 - \epsilon_1} \right)^{q_*} \equiv e^{-\lambda(g) n},
    \end{equation}
    where we used $E = \varepsilon_1 n$ and defined 
    \begin{equation}
    	\lambda(g) = \frac{\epsilon_2-\epsilon_1 - 4g}{2\Delta E} \ln\left(\frac{\epsilon_2-\epsilon_1}{2g}\right).
    \end{equation}
    Finally, we can choose $\Delta E \leq w_0 w_1 \frac{\epsilon_2-\epsilon_1}{\epsilon_2-\epsilon_1 - 4g}$ which finishes the proof. 
\end{proof}

With Lemma~\ref{lem:tail_bound} in hand, we can now bound the numerator in Eq.~\eqref{eq:bottleneck_expander} as follows. 

\begin{theorem}[Stability of states with extensive energy barriers]\label{thm:stability}
    Let $\mathcal{V}$ satisfying Eq.~\eqref{eq:def_lin_barrier} with $\epsilon = E_{0,\min}(\mathcal{V}) / n < \kappa/2$. Then there exists $\beta_* < \frac{\kappa-2\epsilon}{\ln{2}}$ and $g_*(\beta) > 0$ such that for $\beta > \beta_*$ and $g < g_*(\beta)$, the perturbed Hamiltonian $H_0 + V$ has bottleneck ratio
    \begin{equation}
        \Delta = \frac{||P_{\partial\mathcal{V}}\rho_G||_1}{\tr(P_\mathcal{V}\rho )} = e^{-\Omega(n)}.
    \end{equation}
\end{theorem}

\begin{proof}

We have already bounded the denominator of the bottleneck ratio in Eq.~\eqref{eq:weight_in_V}. We will now use Lemma~\ref{lem:tail_bound} to bound the numerator, assuming $g \leq \kappa / 8$. To do so, consider
a state $\ket{\varphi} \in \mathcal{\partial V}$ which is guaranteed to be in the subspace $\hilbert_{>}$ relative to eigenstates $\ket{i}$ of $H$ with energy $E_i \leq (\epsilon + \kappa/2)n$. We can thus bound the expectation value of $e^{-\beta H}$ in this case by expanding it in the eigenbasis and splitting the sum into its low- and high-energy parts:
\begin{align}\label{eq:split_sum}
    \langle \varphi |e^{-\beta H}|\varphi\rangle &= \sum_{i: E_i \leq (\epsilon+\kappa /2)n} e^{-\beta E_i} |\langle i |\varphi\rangle|^2 + \sum_{i: E_i > (\epsilon+\kappa /2)n} e^{-\beta E_i} |\langle i |\varphi\rangle|^2  \nonumber \\
    &\leq \sum_{i: E_i \leq (\epsilon+\kappa /2)n} e^{-\beta E_i} e^{-\lambda_\kappa(g)n} + \sum_{i: E_i > (\epsilon+\kappa /2)n} e^{-\beta (\epsilon+\kappa/2)n} |\langle i |\varphi\rangle|^2 \nonumber \\
    &\leq 2^n e^{-\lambda_\kappa(g) n} + e^{-\beta (\epsilon+\kappa/2)n}, 
\end{align}
where $\lambda_\kappa(g) \geq \frac{\kappa}{4w_0w_1}\ln\left(\frac{\kappa}{4g}\right)$, with $w_0,w_1$ being the constants defined in Lemma~\ref{lem:tail_bound}.

To bound $\tr(P_{\partial\mathcal{V}}e^{-\beta H})$ we can apply Eq.~\eqref{eq:split_sum} for each term, which yields an additional factor of $2^n$.  The partition function is also upper bounded by $Z \leq 2^n$. Combining all these and using Eq.~\eqref{eq:bottleneck_expander} we get an upper bound on the bottleneck ratio
\begin{align}
    \Delta^2 \leq (8^n e^{-\lambda_\kappa(g) n} + 4^n e^{-\beta (\epsilon+\kappa/2)n}) e^{2\beta (\epsilon + g) n} = e^{(3\ln{2}-\lambda_\kappa(g)+2\beta (\epsilon + g))n} + e^{(2\ln{2} + \beta (\epsilon + 2g - \kappa/2))n}.
\end{align}
We can ensure that the right hand side is exponentially small in $n$ provided that the following two conditions are met:
\begin{align}
    \beta\left(\frac{\kappa}{2} - \epsilon - 2g\right) &> 2\ln{2}, \\
    \lambda_\kappa(g) - 2\beta(\epsilon + g) &> 3\ln{2}.
\end{align}
For $\frac{\kappa-2\epsilon}{\ln{2}} < \beta < \infty$ we can always choose $g$ small enough to ensure that both of these conditions are satisfied. For the first, we need $4g < \kappa - 2\epsilon - \frac{1}{\beta} \ln{2}$. For the second, it is sufficient to satisfy
\begin{equation}
    \frac{\kappa}{4w_0w_1} \ln\left(\frac{\kappa}{4g}\right) > 3\ln{2} + \beta(\kappa/2 + \epsilon),
\end{equation}
which is satisfied by taking 
\begin{equation}\label{eq:g_bound}
    g < \frac{\kappa}{4} e^{-4w_0w_1 (\beta(\kappa/2+\epsilon) + 3\ln{2}) / \kappa}.
\end{equation}

\end{proof}

\begin{remark}
    A few closing remarks are in order about Thm.~\ref{thm:stability}. First, note that the restriction on $\epsilon < \kappa/2$ could easily be relaxed to $\epsilon < \kappa$ by appropriately adjusting the various constant appearing in the proof. Thus the condition is that the minimal energy density in $\partial\mathcal{V}$ is at least twice as that of $\mathcal{V}$. This condition was brought about because of the appearance of the square root in the numerator of Eq.~\eqref{eq:bottleneck_expander}, which in turn corresponds to the worst case scenario, where the matrix elements of $\rho_G$ that are off-block-diagonal with respect to the decomposition $\hilbert = \partial\mathcal{V} \oplus \partial\mathcal{V}^\perp$ are the same size as the matrix elements within the block $\partial\mathcal{V}$. In fact, $H_0$ is block-diagonal with respect to this decomposition so the off-diagonal contributions go to zero in the limit $g\to 0$. We thus expect that the proof can be strengthened and possibly the restriction on $\epsilon$ removed completely. 

    Similarly, the bound Eq.~\eqref{eq:g_bound} on $g$ gets weaker with decreasing temperature, and goes to zero in the zero-temperature limit $\beta\to\infty$.  This contrast with physical intuition that suggests that stability should be enhanced at lower temperatures, and we indeed expect that it should be possible to keep $g$ finite in the zero temperature limit without invalidating the conclusions of the theorem.  
    
\end{remark}

\bibliography{references.bib}